\documentclass[11pt,a4paper]{article}

\usepackage[utf8]{inputenc}
\usepackage{smallstyle}

\hypersetup{
  pdftitle={Tight Bounds for Repeated Balls-into-Bins},
  pdfauthor={Dimitrios Los, Thomas Sauerwald}
  }

\allowdisplaybreaks

\begin{document}

\title{Tight Bounds for Repeated Balls-into-Bins}

\author[]{Dimitrios Los\thanks{\texttt{dimitrios.los@cl.cam.ac.uk}} }
\author[]{Thomas Sauerwald\thanks{\texttt{thomas.sauerwald@cl.cam.ac.uk}}}
\affil[]{Department of Computer Science \& Technology, University of Cambridge}

\maketitle

\begin{abstract}
We study the \emph{repeated balls-into-bins} process introduced by Becchetti, Clementi, Natale, Pasquale and Posta~\cite{BCNPP19}. 
This process starts with $m$ balls arbitrarily distributed across $n$ bins. At each round $t=1,2,\ldots$, one ball is selected from each non-empty bin, and then placed it into a bin chosen independently and uniformly at random.  We prove the following results:
\begin{itemize}
    \item For any $n \leq m \leq \mathrm{poly}(n)$, we prove a lower bound of $\Omega(m/n \cdot \log n)$ on the maximum load. For the special case $m=n$, this matches the upper bound of $\Oh(\log n)$, as shown in~\cite{BCNPP19}. It also provides a positive answer to the conjecture in~\cite{BCNPP19} that for $m=n$ the maximum load is $\omega(\log n/ \log \log n)$ at least once in a polynomially large time interval. For $m \in [\omega(n), n \log n]$, our new lower bound disproves the conjecture in~\cite{BCNPP19} that the maximum load remains $\Oh(\log n)$.
    \item For any $n \leq m \leq \mathrm{poly}(n)$, we prove an upper bound of $\Oh(m/n \cdot \log n)$ on the maximum load for all steps of a polynomially large time interval. This matches our lower bound up to multiplicative constants.
    \item For any $m \geq n$, our analysis also implies an $\Oh( m^2 / n)$ waiting time to reach a configuration with a $\Oh(m/n \cdot \log m)$ maximum load, even for worst-case initial distributions.
    \item For any $m \geq n$, we show that every ball visits every bin in $\Oh(m \log m)$ rounds. For $m = n$, this improves the previous upper bound of $\Oh(n \log^2 n)$ in~\cite{BCNPP19}. We also prove that the upper bound is tight up to multiplicative constants for any $n \leq m \leq \mathrm{poly}(n)$.
\end{itemize}
\end{abstract}

\maketitle

\clearpage

\clearpage
\tableofcontents
~
\clearpage

\section{Introduction}

We consider the allocation processes involving $m$ balls (jobs or data items) to $n$ bins (servers or memory cells), by allowing each ball to choose from a set of randomly chosen bins. The goal is to allocate (or re-allocate) balls efficiently, while also keeping the load distribution balanced. The balls-into-bins framework has found numerous applications in 
hashing, load balancing, routing (we refer to the surveys~\cite{MR01} and~\cite{W17} for more details).

A classical sequential allocation algorithm is the \DChoice process introduced by Azar, Broder, Karlin and Upfal~\cite{ABKU99} and Karp, Richard, Luby, and Meyer auf der Heide~\cite{KLM96}, where for each ball to be allocated, we sample $d \geq 1$ bins uniformly and then place the ball in the least loaded of the $d$ sampled bins. It is well-known that for the \OneChoice process ($d=1$), the maximum load is \Whp\footnote{In general, with high probability refers to probability of at least $1 - n^{-c}$ for some constant $c > 0$.}~$\Theta(\log n/\log \log n)$ for $m = n$ and $m/n +\Theta\bigr( \sqrt{ m/n \cdot \log n } \bigr)$ for $m = \Omega(n \log n)$. In particular, this gap between maximum and average load grows significantly as $m/n \rightarrow \infty$, which is called the {\em heavily loaded case}. For $d=2$,~\cite{ABKU99} proved that the maximum load is only $m/n + \log_2 \log n+\Oh(1)$ for $m=n$. This result was generalized by Berenbrink, Czumaj, Steger and V\"{o}cking~\cite{BCSV06} who proved that the same guarantee also holds for $m \geq n$, in other words, even as $m/n \rightarrow \infty$, the difference between the maximum and average load remains a slowly growing function in $n$ that is independent of $m$.
This improvement of \TwoChoice over \OneChoice has been widely known as the ``power of two choices''.

In this work, we investigate the \textit{repeated balls-into-bins (RBB)} process, introduced by Becchetti, Clementi, Natale, Pasquale and Posta~\cite{BCNPP19}. In this process,  there are $m$ balls initially allocated arbitrarily across $n$ bins. In each round, one ball is removed from each non-empty bin and then each of these balls is allocated to one bin sampled uniformly at random (see \cref{fig:rbb_example}). This setting differs from the classical balls into bins setting in that the number of balls is fixed and the amount of balls we re-allocate in each round varies from $1$ to $n$. Unlike \TwoChoice (or \DChoice), this re-allocation is performed without inspecting the load of any bin or taking additional samples.

Becchetti et al.~\cite{BCNPP19} proved that for $m = n$, starting from an arbitrary configuration, \Whp~after $\Oh(n)$ rounds, the process reaches a maximum load of $\Oh(\log n)$ and remains in such a configuration for $\poly(n)$ rounds. Thus, the RBB process is a natural instance of a self-stabilizing system, and falls into a long line of research on random-walk based algorithms for stabilization and consensus~\cite{BCNPT16,C11,HP01,IJ90,PU89}. More recently, Cancrini and Posta~\cite{CP20} proved that the mixing time is $\Oh(L)$ where $L$ is the maximum load at the initial configuration. 

\textbf{Our Results.} In this work, we settle two conjectures stated in~\cite{BCNPP19} and prove tight bounds for the more general case with $n \leq m \leq \poly(n)$.

Becchetti et al.~\cite{BCNPP19} conjectured that the $\Oh(\log n)$ upper bound holds for all $m = \Oh(n \log n)$. They also conjectured that for $m = n$, the maximum load is $\omega(\log  n/\log \log n)$. We resolve both conjectures, proving an $\Omega(m/n \cdot \log n)$ lower bound on the maximum load \Whp~in any interval of length $\Omega(m^2/n^2 \cdot \log^4 n)$ and for any $n \leq m \leq \poly(n)$ (\cref{lem:lower_bound}). This disproves the first conjecture, but confirms the second one, showing that for $m = n$, the maximum load is \Whp~$\Theta(\log n)$.

For the case $m \geq n$, we also prove that starting from an arbitrary configuration after $\Oh(m^2/n)$ rounds, \Whp~we reach a configuration with a maximum load of $\Oh(m/n \cdot \log m)$ (\cref{lem:convergence}). For $n \leq m \leq \poly(n)$, we show that the process stabilizes in such a configuration there for at least $m^2$ rounds (\cref{thm:stabilization}). 

Becchetti et al.~\cite{BCNPP19} also studied the \textit{cover time} (or \textit{traversal time}) of a ball, which is the time required to visit all $n$ bins. For $m = n$, they proved an $\Oh(n \log^2 n)$ bound on the traversal time. For any $m \geq n$, we improve this to $\Oh(n \log m)$, and also show that it is tight up to constant factors for any $m=\poly(n)$ (\cref{pro:traversal}).

\textbf{Intuition and Techniques.} For the upper bound we use an exponential potential $\Phi$ with smoothing parameter $\Theta(n/m)$. Provided that $\Phi$ is $\poly(m)$, we immediately obtain the $\Oh(m/n \cdot \log m)$ bound on the maximum load. Our analysis exploits that after only $\Oh((m/n)^2)$ rounds, sufficiently many bins will become empty, which in turn will reduce the number of balls being re-allocated. This then helps to reduce the load of any non-empty bin, since these are guaranteed to lose one ball per round, but only receive in expectation less than one ball in total from the other non-empty bins. As we will prove, the actual equilibrium will have most bins being empty roughly every $\Oh(m/n)$ rounds. To establish this, we employ some martingale and drift-arguments to first prove that any bin which starts at load $\Oh(m/n)$, becomes empty after $\Oh( (m/n)^2 )$ rounds with constant probability $>0$. Secondly, we prove that if this happens to a fixed bin, the empty load state will be revisited $\Omega(m/n)$ times during the next $\Oh( (m/n)^2)$ rounds. In some sense, this is a generalization of the approach in~\cite{BCNPP19}, where they also bounded the fraction of empty bins for the case $m = n$.

A kind of reversed  argument is used for the lower bound. Here, the goal is to prove that each bin is only empty every $\Oh(m/n)$ rounds on average. This shows that the RBB process can be approximated by a \OneChoice process where at least an $1 - \Oh(n/m)$ fraction of the balls are allocated. For $t = \Omega(m^2/n^2 \cdot \log n)$, this yields a maximum load of $\Omega(m/n \cdot \log n)$. To prove that bins are not empty ``too often'', we establish a link between a quadratic potential and the number of empty bins, similar to that in \cite[Lemma 6.2]{LSS21}. This connection essentially implies that whenever the fraction of empty bins is $\omega(n/m)$, then the quadratic potential decreases. By aggregating sufficiently over many rounds, we can conclude that, on average, the number of empty bins cannot be too large.

\textbf{Further Related Work.} Cancrini and Posta investigated the behavior of the RBB process for a large number of rounds, 
and established ``propagation of chaos''~\cite{CP19}, meaning that under some conditions on the initial load distribution, the load of the bins become eventually independent. In~\cite{CP19}, the authors prove results for the RBB process considered here, while~\cite{CP21} considered more general re-allocation rules.
Another variant of the RBB setting was studied in~\cite{BFKMNW18}, where in each round one ball is deleted from each bin and an expected $\lambda n$ new balls arrive and are distributed in parallel to the bins. In contrast to the RBB model, this means that the number of balls in the system is not fixed.

The RBB is an instance of a discrete time closed Jackson network~\cite{J04,K76}. However, in RBB, updates are happening synchronously and in parallel, while in most queuing models updates occur asynchronously based on independent point processes. As also pointed out in~\cite{CP19,CP21}, this leads to a non-reversible Markov Chain, which seems to make the computation of the stationary distribution intractable. Furthermore, formal methods have been used to prove guarantees for RBB with $m = n$~\cite{BGHT22}. The RBB setting has also been applied to analyze protocols in short packet communications~\cite{YYWJ21}.

Czumaj, Riley and Scheideler~\cite{CRS03} studied a similar re-allocation process where in each round one random ball is allocated to a random of $d$ bin choices. These are also related to randomized rerouting protocols studied in~\cite{BFGGHM07,BHT11}. In another parallel allocation processes, Berenbrink, Czumaj, Englert, Friedetzky and Nagel~\cite{BCE12} proved an $\Oh(\log n)$ gap for the \TwoChoice process where balls are allocated in batches of $n$ balls and was recently improved to $\Oh(\log n / \log \log n)$ in~\cite{LS22noise}.

\textbf{Organization.} In \cref{sec:notation} we introduce some standard balls-into-bins notations and define the processes.  In \cref{sec:lower_bounds}, we prove our lower bound on the maximum load. In \cref{sec:upper_bounds}, we prove an upper bound on maximum load and also analyze the time until such configuration is reached and preserved (convergence time). In \cref{sec:traversal}, we analyze the traversal time. In \cref{sec:experiments}, we present some empirical results on the RBB process. We conclude the paper with a summary and a few open problems in \cref{sec:conclusions}.

\section{Notation and Definitions} \label{sec:notation}

We consider a set of $n$ bins labeled $[n]:=\left\{1,2,\ldots,n \right\}$.
By $x^{t}$ we denote the $n$-dimensional \textit{load vector} after $t$ rounds, and $x^{0}$ is the initial load vector. In our processes, no balls are added or removed, and the existing $m$ balls are only re-allocated; hence, $\sum_{i=1}^n x_i^t = m$ for all $t \geq 0$. 

By $F^t := \left| \left\{ i \in [n] \colon x_i^{t} = 0 \right\} \right|$ we denote the number of \textit{empty (free) bins} and by $f^t := \frac{1}{n} \cdot F^t$ the fraction of empty bins. Similarly $\kappa^t := n - F^t$ is the number of \textit{non-empty bins}. Since it will be important to track the number of empty bins over a time interval, we also define $F_{t_0}^{t_1}$ as the total number of pairs of empty bins and rounds in the entire interval $[t_0, t_1]$, i.e.,
\[
F_{t_0}^{t_1} := \sum_{t = t_0}^{t_1} F^t.
\]

\newcommand{\bin}[2]
{
 \draw [thick, rounded corners=6pt] (#1-0.35,#2+5) to ++(0,-5) to ++(0.7,0) to ++(0,5); 
}

\newcommand{\HBall}[4]
{
 \node (#4) [fill=yellow,minimum size=.6cm,circle,draw=red,ultra thick] at (#1,#2-0.45) {#3};
}

\newcommand{\Ball}[4]
{
 \node (#4) [fill=yellow,minimum size=.6cm,circle,draw=black,thin] at (#1,#2-0.45) {#3};
}

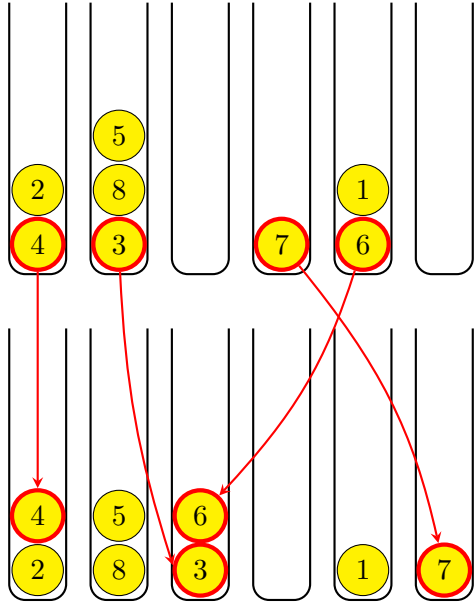
\begin{figure}
\begin{center}
\begin{tikzpicture}[xscale=1.07,yscale=0.72]

\bin{0}{0}
\bin{1}{0}
\bin{2}{0}
\bin{3}{0}
\bin{4}{0}
\bin{5}{0}

\HBall{0}{1}{4}{4}
\Ball{0}{2}{2}{2}
\HBall{1}{1}{3}{3}
\Ball{1}{2}{8}{8}
\Ball{1}{3}{5}{5}
\HBall{3}{1}{7}{7}
\HBall{4}{1}{6}{6}
\Ball{4}{2}{1}{1}

\begin{scope}[yshift=-6cm]

\bin{0}{0}
\bin{1}{0}
\bin{2}{0}
\bin{3}{0}
\bin{4}{0}
\bin{5}{0}

\HBall{0}{2}{4}{4n}
\HBall{2}{1}{3}{3n}
\HBall{5}{1}{7}{7n}
\HBall{2}{2}{6}{6n}

\Ball{0}{1}{2}{2}

\Ball{1}{1}{8}{8}
\Ball{1}{2}{5}{5}

\Ball{4}{1}{1}{1}

\end{scope}

\draw[-stealth, red, thick] (4) to (4n);
\draw[-stealth, red, thick, bend right=-10] (6) to (6n);
\draw[-stealth, red, thick, bend right=5] (3) to (3n.170);
\draw[-stealth, red, thick, bend right=-10] (7) to (7n);

\end{tikzpicture}
\caption{Illustration of one step of RBB with $m=8$ balls and $n=6$ bins. The balls highlighted in red are re-allocated to bins chosen randomly among $\{1,2,\ldots,6\}$.}
\label{fig:rbb_example}
\end{center}
\end{figure}

\begin{framed}
\vspace{-.45em} \noindent
\underline{RBB (Repeated Balls-into-Bins Process):} \\
\textsf{Iteration:} At each round $t = 1,2,\ldots$
\begin{itemize}[topsep=0pt]
    \item For each of the $\kappa^t = n-F^t$ non-empty bins, remove one ball and re-allocate it to a bin chosen independently and uniformly at random among $[n]$.
\end{itemize}
\end{framed}

More specifically, in each round we choose $\kappa^t$ bins $z_{1}^t, \ldots , z_{\kappa^t}^t \in [n]$ uniformly at random and the load vector at step $t+1$ is given by
\[
  x_i^{t+1} := x_i^t - \mathbf{1}_{x_i^t > 0}  + \sum_{j = 1}^{\kappa^t} \mathbf{1}_{z_j^t = i}, \quad \text{for each }i \in [n].
\]

Hence, we can express the marginal load distribution of an arbitrary bin $i \in [n]$ at round $t \geq 0$ (i.e., having completed $t$ iterations before), as
\begin{align}
 x_i^{t+1} = x_i^{t}  - \mathbf{1}_{x_i^t > 0} +  \Bin( \kappa^t,1/n), \label{eq:marginal}
\end{align}
where with slight abuse of notation, we write $\Bin( \kappa^t,1/n)$ as a placeholder for a random variable (independent of $\mathfrak{F}^t$, the entire history of the process up to round $t$) which has distribution $\Bin( \kappa^t,1/n)$.

Similarly, assuming each bin acts as a FIFO queue on the incoming and departing balls, we can follow the trajectory of an arbitrary single ball. Only if the ball is at the front of its queue, it will be re-allocated to a bin chosen randomly from $[n]$ in the next round. A natural question is the so-called \emph{cover time} (or \emph{traversal time}), the expected time until every ball has been allocated to each bin~\cite{BCNPP19}. This is related to the well-studied \emph{cover time} of parallel random walks on graphs, but with the constraint that only one walk can leave each vertex (=bin) at a time.

\section{Lower Bound on the Maximum Load for \texorpdfstring{$n \leq m \leq \poly(n)$}{n <= m <= poly(n)}} \label{sec:lower_bounds}

In this section we prove a lower bound on the maximum load of $\Omega(m/n \cdot \log n)$ which holds \Whp~every $\Oh(m^2/n^2 \cdot \log^4 n)$ steps for any $n \leq m \leq \poly(n)$. This matches the upper bound of \cref{sec:stabilization} up to multiplicative constants.

\subsection{Quadratic Potential and Empty Bins}

We now define the \textit{quadratic potential function}, as
\[
\Upsilon^t := \sum_{i = 1}^n (x_i^t)^2.
\]

We prove an important relation between the quadratic potential $\Upsilon^t$ and the number of empty bins $F^t$. These relations are similar to the ones used in~\cite{LSS21} to show that the absolute value potential is small in a constant fraction of the rounds. The key insight is that the quadratic potentials drops in expectation as soon as the fraction of empty bins is of order $\Omega(n/m)$. This will be crucial in the derivation of our lower bounds.

\begin{lem} \label{lem:quadratic_empty} \label{lem:quadratic_drop}
Consider the RBB setting with any $m \geq 1$. Then, for any round $t \geq 0$,
\begin{align*}
\Ex{\left. \Upsilon^{t+1} \,\right|\, \mathfrak{F}^t}
 & \leq \Upsilon^t - 2 \cdot \frac{m}{n} \cdot F^t + 2n.
\end{align*}
\end{lem}
\begin{proof}
Let us define the binomial random variable $Z \sim \mathrm{Bin}(\kappa^t,  \frac{1}{n})$. For any bin $i \in [n]$ with load $x_i^t \geq 1$,
\begin{align*}
\Ex{\left. \Upsilon_i^{t+1} \,\right|\, \mathfrak{F}^t}
 & = \sum_{z = 0}^{\kappa^t} (x_i^t + z - 1)^2 \cdot \binom{\kappa^t}{z} \cdot \frac{1}{n^z} \cdot \Big(1 - \frac{1}{n} \Big)^{\kappa^t - z} \\
 & = (x_i^t)^2 \cdot \sum_{z = 0}^{\kappa^t} \binom{\kappa^t}{z} \cdot \frac{1}{n^z} \cdot \Big(1 - \frac{1}{n} \Big)^{\kappa^t - z} \\
 & \quad \quad + 2 \cdot x_i^t \cdot \sum_{z = 0}^{\kappa^t} (z - 1) \cdot \binom{\kappa^t}{z} \cdot \frac{1}{n^z} \cdot \Big(1 - \frac{1}{n} \Big)^{\kappa^t - z} \\
 & \quad \quad + \sum_{z = 0}^{\kappa^t} (z - 1)^2 \cdot \binom{\kappa^t}{z} \cdot \frac{1}{n^z} \cdot \Big(1 - \frac{1}{n} \Big)^{\kappa^t - z} \\
 & = (x_i^t)^2 \cdot \Ex{Z} + 2 \cdot x_i^t \cdot \Ex{Z - 1} + \Ex{(Z - 1)^2} \\
 & \stackrel{(a)}{=} (x_i^t)^2 + 2 \cdot x_i^t \cdot \Big(\frac{\kappa^t}{n} - 1\Big) + \kappa^t \cdot (\kappa^t - 1)\cdot \frac{1}{n^2} - \frac{\kappa^t}{n} + 1 \\
 & \leq (x_i^t)^2 + 2 \cdot x_i^t \cdot \Big(\frac{\kappa^t}{n} - 1\Big) + 2,
\end{align*}
having used in $(a)$ that $\Ex{Z} = \frac{\kappa^t}{n}$ and $\Ex{Z^2} = \kappa^t \cdot \frac{1}{n} \cdot (1-\frac{1}{n}) + (\kappa^t)^2 \cdot ( \frac{1}{n})^2$, and thus 
\[
\Ex{(Z-1)^2} =\kappa^t \cdot  \frac{1}{n} \cdot \left(1- \frac{1}{n} \right) + (\kappa^t)^2 \cdot \left( \frac{1}{n} \right)^2 - 2 \cdot  \frac{\kappa^t}{n} + 1 = \kappa^t \cdot (\kappa^t - 1) \cdot \frac{1}{n^2} - \frac{\kappa^t}{n} + 1.
\]
Similarly for an empty bin $i \in [n]$ with $x_i^t=0$, the contribution is
\begin{align*}
\Ex{\left. \Upsilon_i^{t+1} \,\right|\, \mathfrak{F}^t}
 & = \sum_{z = 0}^{\kappa^t} z^2 \cdot \binom{\kappa^t}{z} \cdot \frac{1}{n^z} \cdot \Big(1 - \frac{1}{n} \Big)^{\kappa^t - z} = \frac{\kappa^t}{n} + \frac{\kappa^t \cdot (\kappa^t -1 )}{n^2}.
\end{align*}
Hence, by aggregating the contributions of the $\kappa^t$ bins non-empty bins and the $n - \kappa^t$ empty bins we obtain
\begin{align*}
\Ex{\Upsilon^{t+1} \mid \mathfrak{F}^t}
 &\leq \Upsilon^t + \sum_{i \in [n] \colon x_i^t \geq  1 } \left( 2 \cdot x_i^t \cdot \left( \frac{\kappa^t}{n} -1 \right) + 2 \right) 
 + \sum_{i \in [n] \colon x_i^t = 0 }  \left( \frac{\kappa^t}{n} + \frac{\kappa^t \cdot (\kappa^t -1 )}{n^2} \right)
 \\
 &\leq \Upsilon^t + \left( \frac{\kappa^t}{n} - 1 \right) \cdot 2 \cdot m + 2 \kappa^t 
 + (n-\kappa^t) \cdot 2 \\
 & = \Upsilon^t - 2 \cdot \frac{m}{n} \cdot F^t + 2n,
\end{align*}
where in the last inequality we used that $\kappa^t \leq n$, which finishes the proof.
\end{proof}

\subsection{Upper Bounding the Number of Empty Bins}

The key insight is that the quadratic potential drops in expectation as soon as the fraction of empty bins is of order $\Omega(n/m)$. This is crucial to upper bound the number of empty bins in an interval. This relation is similar to the ones used in~\cite{LS22noise,LSS21}, where an interplay between the quadratic potential and the absolute value potential was used to show that the absolute value potential is small in a constant fraction of the rounds.

The next lemma shows that for any sufficiently long interval, either there is a maximum load that is $\Omega(m/n \cdot \log n)$ or the fraction of empty bins in the interval is $\Oh(n/m)$. Note that we indeed need the interval to be long enough as starting with the perfectly balanced load vector may require several rounds to reach a gap of $\Omega(m / n \cdot \log n)$ even for the \OneChoice process. 

\begin{lem} \label{lem:quadratic_eps}
Consider the RBB process with any $n \leq m \leq n^k$ for some constant $k \geq 1$ and any $1 \leq \hat{c} \leq n$. Then, for any round $t_0 \geq 0$ and  $t_1 := t_0 + \hat{c} \cdot \big( \frac{m}{n} \cdot \log n \big)^2$,
\begin{align*}
\Pro{ \left. \left\lbrace F_{t_0}^{t_1} < \frac{n^2}{4m} \cdot (t_1 - t_0 + 1) \right\rbrace \cup \bigcup_{t \in [t_0, t_1]} \left\lbrace \max_{i \in [n]} x_i^t > \frac{m}{n} \cdot \log n\right\rbrace ~\right|~ \mathfrak{F}^{t_0}} 
&\geq 1 - e^{ - \frac{\hat{c}}{18} }.
\end{align*}
\end{lem}

\begin{proof}
We define for any $t \geq t_0$ the sequence
\[
 Z^{t} := \Upsilon^{t} - 2 \cdot (t-t_0) \cdot n + 2 \cdot \frac{m}{n} \cdot F_{t_0}^{t-1},
\]
where $F_{t_0}^{t_0-1}=0$.
This sequence forms a super-martingale since by \cref{lem:quadratic_drop},
\begin{align*}
\Ex{\left. Z^{t+1} \,\right|\, \mathfrak{F}^{t}} 
  & = \Ex{\Upsilon^{t+1} - 2 \cdot (t-t_0+1) \cdot n + 2 \cdot \frac{m}{n} \cdot F_{t_0}^{t} ~\Big|~ \mathfrak{F}^{t}} \\
  & = \Ex{\left. \Upsilon^{t+1} \,\right|\, \mathfrak{F}^t} - 2 \cdot (t-t_0+1) \cdot n + 2 \cdot \frac{m}{n} \cdot F_{t_0}^{t} \\
  & \leq \Upsilon^{t} + 2 \cdot n - 2 \cdot \frac{m}{n} \cdot F^t - 2 \cdot (t-t_0+1) \cdot n + 2 \cdot \frac{m}{n} \cdot F_{t_0}^{t} \\
  & = \Upsilon^t - 2 \cdot (t-t_0 ) \cdot n + 2 \cdot \frac{m}{n} \cdot F_{t_0}^{t-1} \\
  &= Z^{t}.
\end{align*}

Further, let $\tau:=\min\{ t \geq t_0 \colon \max_{i \in [n]} x_i^{t} > \frac{m}{n} \cdot \log n \}$ and consider the stopped random variable
\[
 \tilde{Z}^{t} := Z^{t \wedge \tau},
\]
which is then also a super-martingale.

To prove concentration of $\tilde{Z}^{t}$, we will now derive an upper bound on $\left| \tilde{Z}^{t+1} - \tilde{Z}^{t} \right|$ conditional on $\mathfrak{F}^{t}$.
\medskip 

\noindent\textbf{Case 1} [$t \geq \tau$].
In this case, $\tilde{Z}^{t+1} = Z^{(t+1) \wedge \tau} = Z^{\tau}$, and similarly, $\tilde{Z}^{t} = Z^{t \wedge \tau} = Z^{\tau}$, so $| \tilde{Z}^{t+1} - \tilde{Z}^{t}|=0$.
\medskip

\noindent\textbf{Case 2} [$t < \tau$]. Hence for $t$ we have $\max_{i \in [n]} x_i^t \leq \frac{m}{n} \cdot \log n$ and thus \cref{lem:bound_upsilon_change} implies that the biggest change in the quadratic potential is \Whp~at most $2 \cdot m \cdot \log n$ and under this condition,
\[
|\tilde{Z}^{t+1} - \tilde{Z}^t| \leq 2 \cdot m \cdot \log n + 2 n + 2 \cdot \frac{m}{n} \cdot n \leq 3 \cdot m \cdot \log n.
\]
Combining the two cases above, we conclude,
\[
\Pro{\bigcap_{t \in [t_0, t_1 - 1]} \left\{ |\tilde{Z}^{t+1} - \tilde{Z}^t| \leq 3 \cdot m \cdot \log n \right\} } \geq 1 - n^{-\omega(1)} \cdot (t_1 - t_0) \geq 1 - n^{-\omega(1)},
\]
since $t_1 - t_0 \leq \poly(n)$.

Using the concentration inequality \cref{thm:chu_lu_thm_8_3} with bad event,
$\mathcal{B}^t := \neg \bigcap_{t \in [t_0, t]} \{ |\tilde{Z}^{t+1} - \tilde{Z}^t| \leq 3 \cdot m \cdot \log n \}
$ and $\lambda = \hat{c} \cdot \frac{m^2}{n} \cdot \log^2 n$, we get
\begin{align*}
 \Pro{ \tilde{Z}^{t_1+1} - \tilde{Z}^{t_0} > \lambda} 
 &\leq \exp\left( - \frac{ \lambda^2 }{ 2 \cdot \sum_{t=t_0}^{t_1} (3 \cdot m \cdot \log n)^2  } \right) + \Pro{\mathcal{B}} \\ 
 &= \exp\left( - \frac{ \hat{c}^2 \cdot \Big(\frac{m^2}{n} \cdot \log^2 n \Big)^2 }{ 18 \cdot \hat{c} \cdot \big( \frac{m}{n} \cdot \log n \big)^2 \cdot (m \cdot \log n)^2 } \right) + \Pro{\mathcal{B}}\\
 &\leq e^{ - \frac{\hat{c}}{18} } + n^{-\omega(1)} \leq 2 \cdot e^{ - \frac{\hat{c}}{18} },
\end{align*}
Thus,  \[\Pro{ \left\lbrace Z^{t_1+1}  \leq Z^{t_0} + \lambda \right\rbrace \cup \bigcup_{t \in [t_0, t_1]} \left\lbrace \max_{i \in [n]} x_i^t \geq \frac{m}{n} \cdot \log n\right\rbrace} < 1 - 2 \cdot e^{ - \frac{\hat{c}}{18} }.
\]
For the sake of a contradiction, assume now that
\[
  F_{t_0}^{t_1} \geq \frac{4n^2}{m} \cdot (t_1 - t_0 + 1).
\]
If $Z^{t_1+1} \leq Z^{t_0} + \lambda$ holds, then we have
\[
\Upsilon^{t_1+1}  - 2 \cdot (t_1-t_0+1) \cdot n + 2 \cdot \frac{m}{n} \cdot F_{t_0}^t \leq \Upsilon^{t_0} + \lambda .\]Rearranging the inequality above gives
\begin{align} \Upsilon^{t_1+1} & \leq \Upsilon^{t_0} + \lambda + 2 \cdot (t_1-t_0+1) \cdot n - 2 \cdot \frac{m}{n} \cdot F_{t_0}^t \notag  \\
  & \leq \Upsilon^{t_0} + \lambda + 2 \cdot (t_1-t_0+1) \cdot n - 8 \cdot n \cdot (t_1 - t_0+1)\notag \\
  & \leq \Upsilon^{t_0} + \lambda - 6 \cdot (t_1 - t_0 +1) \cdot n.\label{eq:upsilon_contradiction}
\end{align} 
Recall that we start from a round $t_0$ where $\max_{i \in [n]} x_i^{t_0} \leq \frac{m}{n} \cdot \log n$, and therefore also $\Upsilon^{t_0} \leq n \cdot \big( \frac{m}{n} \cdot \log n \big)^2$. Thus, by \eqref{eq:upsilon_contradiction} we have   
\[\Upsilon^{t_1+1} \leq n \cdot \Big( \frac{m}{n} \cdot \log n \Big)^2 + \hat{c} \cdot n \cdot \Big( \frac{m}{n} \cdot \log n \Big)^2 - 6 \cdot \hat{c} \cdot \Big( \frac{m}{n} \cdot \log n \Big)^2 \cdot n < 0\]
which is a contradiction for large $n$ since $\hat{c} \geq 1$. We conclude that if $Z^{t_1+1} \leq Z^{t_0} + \lambda$, then $F_{t_0}^{t_1} < \frac{n^2}{4m} \cdot (t_1 - t_0 + 1)$ or the stopping time was reached, i.e.
\begin{align*}
 & \Pro{ \left\lbrace F_{t_0}^{t_1} < \frac{n^2}{4m} \cdot (t_1 - t_0 + 1) \right\rbrace \cup \bigcup_{t \in [t_0, t_1]} \left\lbrace \max_{i \in [n]} x_i^t \geq \frac{m}{n} \cdot \log n\right\rbrace } \geq 1 - 2 \cdot e^{ - \frac{\hat{c}}{18} }. \qedhere
\end{align*}
\end{proof}

\subsection{Completing the Proof of the Lower Bound}

To complete the derivation of the lower bound we need to show that in an interval of length $T = \Theta((m/n \cdot \log n)^2)$ with an $\Oh(n/m)$ fraction of empty bins, the maximum load is $\Omega(m/n \cdot \log n)$. This follows by coupling the allocations of the RBB process in the interval with a \OneChoice process with $T \cdot (1 - \Oh(n/m))$ balls. By the following standard expression, for the maximum load, setting $c := \frac{(1 - \gamma)^2}{200} \cdot \frac{1}{\gamma^2}$ (for $\gamma = \Theta(n/m)$), we get the $\Omega(m/n \cdot \log n)$ lower bound on the maximum load for the RBB setting.

\newcommand{\OneChoiceCnlogn}{
Consider the \OneChoice process for $m = c n \log n$ balls where $c \geq 1/\log n$.
Then, we have
\[
\Pro{ \max_{i \in [n]} y_i^{m} \geq \Big(c + \frac{\sqrt{c}}{10}\Big) \cdot \log n} \geq 1 - n^{-2}.
\]}
{\renewcommand{\thelem}{\ref{lem:one_choice_cnlogn}}
	\begin{lem}[Restated, page~\pageref{lem:one_choice_cnlogn}]
\OneChoiceCnlogn
	\end{lem} }
	\addtocounter{lem}{-1}

Putting the lemmas together, we get the desired lower bound.

\begin{lem} \label{lem:lower_bound}
Consider the RBB process with any $n \leq m \leq n^k$ for some constant $k \geq 1$ and let $\gamma := \frac{n}{4m}$. Then, for any round $t_0 \geq 0$ and for $t_1 := t_0 + \frac{1 - \gamma}{200} \cdot \frac{1}{\gamma^2} \cdot \log^4 n$,
\[
\Pro{\bigcup_{t \in [t_0, t_1]} \left\lbrace \max_{i \in [n]} x_i^t \geq 0.008 \cdot \frac{m}{n} \cdot \log n \right\rbrace} \geq 1 - n^{-1}.
\]
\end{lem}

\begin{proof}
Using \cref{lem:quadratic_eps} (for $k := \frac{1 - \gamma}{200} \cdot 16 \cdot \log^2 n \geq 3 \cdot 18 \cdot \log n$), we have for $t_1 = t_0 + \frac{1 - \gamma}{200} \cdot \frac{1}{\gamma^2} \cdot \log^4 n$,
\begin{align} 
 \Pro{ \left\lbrace F_{t_0}^{t_1} < \frac{n^2}{4m} \cdot (t_1 - t_0 + 1) \right\rbrace \cup \bigcup_{t \in [t_0, t_1]} \left\lbrace \max_{i \in [n]} x_i^t \geq \frac{m}{n} \cdot \log n\right\rbrace } & \geq 1 - n^{-2} \label{eq:lb_union_bound_one}.
\end{align}
Consider the $\log^3 n$ sub-intervals $\mathcal{I}_1, \ldots , \mathcal{I}_{\log^3 n}$ of length $\Delta = \frac{1 - \gamma}{200} \cdot \frac{1}{\gamma^2} \cdot \log n$ with starting points $s_j := t_0 + \Delta \cdot (j-1)$. We also define the events for $j \in [\log^3 n]$,
\[
\mathcal{C}_j := \left\lbrace F_{s_j}^{s_{j} + \Delta} < \frac{n^2}{4m} \cdot \Delta \right\rbrace.
\]

Running the RBB process over the interval $[s_j,s_j+\Delta]$ involves reallocating $\Delta \cdot n - F_{s_j}^{s_j+\Delta}$ balls, meaning that we take at least $\Delta \cdot n - F_{s_j}^{s_j+\Delta}$ samples uniformly at random. So if the event $\mathcal{C}_j$ holds, then we sample in total
\[
 \Delta \cdot n  - \frac{n^2}{4m} \cdot \Delta = (1-\gamma) \cdot \Delta \cdot n =: \overline{m}
\]
bins. Hence these re-allocations correspond to a \OneChoice process with $\overline{m}$ balls into $n$ bins; let us denote its load vector by $y^t$ for any round $t \geq 0$ starting from the empty load configuration.
By \cref{lem:one_choice_cnlogn} with $c := \frac{(1 - \gamma)^2}{200} \cdot \frac{1}{\gamma^2}$, we obtain 
\[
 \Pro{ \max_{i \in [n]} y_i^{\overline{m}} \geq \left( c+ \frac{\sqrt{c}}{10} \right) \cdot \log n } \geq 1-n^{-2}.
\]
Further, note that
\begin{align*}
\max_{i \in [n]} y_i^{\overline{m}} 
   & \geq \left( \frac{(1 - \gamma)^2}{200\gamma^2} + \frac{1}{10} \cdot \sqrt{\frac{(1-\gamma)^2}{200 \gamma^2}}\right) \cdot \log n \\ 
 & \geq \frac{1- \gamma}{200 \gamma^2} \cdot \log n + 0.002 \cdot \frac{\log n}{\gamma}
   = \Delta + 0.002 \cdot \frac{\log n}{\gamma}.
\end{align*}

In $\Delta$ rounds, at most $\Delta$ balls can be removed from any single bin, so for any bin $i \in [n]$,
$
 x_i^{s_j+\Delta} \geq x_i^{s_j} + y_{i}^{\overline{m}} - \Delta \geq y_{i}^{\overline{m}} - \Delta.
$
and hence
\[
 \max_{i \in [n]} x_i^{s_j+ \Delta}  \geq \max_{i \in [n]} y_i^{\overline{m}} - \Delta \geq 0.002 \cdot \frac{\log n}{\gamma}.
\]

Next, we define for any round $t \geq 0$
\[
 \mathcal{E}^{t} := \left\lbrace \max_{i \in [n]} x_i^{t} \geq 0.008 \cdot \frac{m}{n} \cdot \log n \right\rbrace.
\]
We have shown that for any $1 \leq j \leq \log^3 n$,
\[
\Pro{ \mathcal{E}^{s_j + \Delta} \cup \neg \mathcal{C}_j} \geq 1 - n^{-2}.
\]
By taking the union bound over the $\log^3 n$ sub-intervals, we conclude that
\begin{align} 
\Pro{\bigcap_{j \in [\log^3 n]} \left( \left\lbrace \max_{i \in [n]} x_i^{s_j + \Delta} \geq 0.008 \cdot \frac{m}{n} \cdot \log n  \right\rbrace \cup \neg \mathcal{C}_j \right)} &\geq 1 - (\log^3 n) \cdot n^{-2} \label{eq:lb_union_bound_two}.
\end{align}

Assuming that $\big\{ F_{t_0}^{t_1} < \frac{n^2}{4m} \cdot (t_1 - t_0 + 1) \big\}$ holds, then using the pigeonhole principle, at least one of these intervals $j$ satisfies $\mathcal{C}_j$, i.e., $\big\{ \cup_{j \in [\log^3 n]} \mathcal{C}_j \big\}$ holds. Hence, by the union bound of \cref{eq:lb_union_bound_one} and \cref{eq:lb_union_bound_two} we conclude that
\begin{align*}
\Pro{ \bigcup_{t \in [t_0, t_1]} \mathcal{E}^t }
& \geq \Pro{ \bigcup_{j \in [\log^3 n]} \mathcal{E}^{s_j+\Delta} } \\
& \geq 1 - \Pro{\neg \bigcap_{j \in [\log^3 n]} \left( \mathcal{E}^{s_j + \Delta} \cup \neg \mathcal{C}_j \right) \cup \bigcap_{j \in [\log^3 n]} \neg \mathcal{C}_j } \\
& \geq \Pro{ \bigcup_{j \in [\log^3 n]} \left( \mathcal{E}^{s_j + \Delta} \cup \neg \mathcal{C}_j \right)} - \Pro{\bigcap_{j \in [\log^3 n]} \neg \mathcal{C}_j } \\
& \geq 1 - n^{-2} - (\log^3 n) \cdot n^{-2} \geq 1 - n^{-1}. \qedhere
\end{align*}
\end{proof}

\section{Upper Bounds on the Maximum Load and Convergence Time} \label{sec:upper_bounds}

In this section, we present the upper bounds on the maximum load and on the converge time for the RBB process. In \cref{sec:exponential_potential}, we introduce the exponential potential function for the RBB process and prove an upper bound on its expected change. In \cref{sec:very_lightly}, we demonstrate its application to the simpler setting where $m < n$, where deterministically there are many empty bins in every round. In \cref{sec:number_of_empty_bins_lb}, we establish that \Whp~for $m \geq n$, an $\Omega(n/m)$ fraction of the bins are empty. In \cref{sec:convergence_time}, we combine this with refined upper bounds on the exponential potential to show that the RBB converges \Whp~in $\Oh(m^2/n)$ rounds to a configuration with an $\Oh(m/n \cdot \log m)$ maximum load. In \cref{sec:stabilization}, for $n \leq m \leq \poly(n)$, we show that \Whp~it remains in a configuration with an $\Oh(m/n \cdot \log n)$ maximum load for at least $m^2$ rounds. This matches the lower of \cref{sec:lower_bounds} up to multiplicative constants.

\subsection{Exponential Potential} \label{sec:exponential_potential}

In this subsection, we introduce the exponential potential function for the RBB process and prove an upper bound on its expected change, relating to the number of non-empty bins $\kappa^t$ at round $t$. This bound is sufficient to obtain upper bounds for the maximum load in the case $m < n$ (\cref{sec:very_lightly}). In \cref{sec:convergence_time}, we extend this to analyze the case $m \geq n$.

The \textit{exponential potential function} with smoothing parameter $\alpha > 0$, is defined as
\[
 \Phi^t := \Phi^t(\alpha) := \sum_{i = 1}^n \Phi_i^t := \sum_{i=1}^n e^{ \alpha x_i^t},
\]
where $x_i^t$ is the load of bin $i$ at round $t$. Note that when $\Phi^t = \poly(n)$, then $\Gap(t) = \Oh\big( \frac{\log n}{\alpha} \big)$. We will now relate the expected change of $\Phi^t$ over one round with the number of non-empty bins $\kappa^t$ in that round.
\begin{restatable}{lem}{ExpChangeGeneral} \label{lem:exp_change_general}
Consider the RBB process with any $m \geq n$. For the potential $\Phi := \Phi(\alpha)$ with any $\alpha > 0$, it holds for any round $t \geq 0$,
\[
\Ex{\left. \Phi^{t+1} \,\right|\, \mathfrak{F}^t} \leq \Phi^t \cdot e^{-\alpha} \cdot e^{\frac{e^{\alpha} - 1}{n} \cdot \kappa^t} + (n - \kappa^t) \cdot e^{\frac{e^{\alpha} - 1}{n} \cdot \kappa^t}.
\]
\end{restatable}
\begin{proof}
Consider the expected contribution of a bin $i \in [n]$ with $x_i^t \geq 1$,
\begin{align*}
\Ex{\left. \Phi_i^{t+1} \,\right|\, \mathfrak{F}^t}
 & = \sum_{z = 0}^{\kappa^t} e^{\alpha (x_i^t + z - 1)} \cdot \binom{\kappa^t}{z} \cdot \left(\frac{1}{n}\right)^z \cdot \left(1 - \frac{1}{n} \right)^{\kappa^t - z} \\
 & = \Phi_i^t \cdot e^{-\alpha} \cdot \sum_{z = 0}^{\kappa^t}  \binom{\kappa^t}{z} \cdot \left(\frac{e^{\alpha}}{n}\right)^z \cdot \left(1 - \frac{1}{n} \right)^{\kappa^t - z} \\
 & \stackrel{(a)}{=} \Phi_i^t \cdot e^{-\alpha} \cdot \left(1 - \frac{1}{n} + \frac{e^{\alpha}}{n}\right)^{\kappa^t} \\
 & \stackrel{(b)}{\leq} \Phi_i^t \cdot e^{-\alpha} \cdot e^{\frac{e^{\alpha} - 1}{n} \cdot \kappa^t},
\end{align*}
using in $(a)$ the binomial identity $\sum_{z = 0}^k \binom{k}{z} p^z q^{k-z} = (p + q)^k$ and in $(b)$ that $1 + z \leq e^z$ for any $z \geq 0$.

For an empty bin $i \in [n]$, its expected contribution is 
\[
\Ex{\left. \Phi_i^{t+1} \,\right|\, \mathfrak{F}^t}
 = \sum_{z = 0}^{\kappa^t} \binom{\kappa^t}{z} \cdot e^{\alpha z} \cdot \left(\frac{1}{n}\right)^z \cdot \left(1 - \frac{1}{n} \right)^{\kappa^t - z} = \left(1 - \frac{1}{n} + \frac{e^{\alpha}}{n}\right)^{\kappa^t} \leq e^{\frac{e^{\alpha} - 1}{n} \cdot \kappa^t}.
\]
Aggregating over all bins, we have
\begin{align*}
\Ex{\left. \Phi^{t+1} \,\right|\, \mathfrak{F}^t} 
 & = \sum_{i : x_i^t \geq 1} \Ex{\Phi_i^{t+1} \mid \mathfrak{F}^t} + \sum_{i : x_i^t = 0} \Ex{\left. \Phi_i^{t+1} \,\right|\, \mathfrak{F}^t}\\
 & \leq \Phi^t \cdot e^{-\alpha} \cdot e^{\frac{e^{\alpha} - 1}{n} \cdot \kappa^t} + (n - \kappa^t) \cdot e^{\frac{e^{\alpha} - 1}{n} \cdot \kappa^t}. \qedhere
\end{align*}
\end{proof}

\subsection{Maximum Load Upper Bound for \texorpdfstring{$m < n$}{m < n}} \label{sec:very_lightly}

In this subsection, we will investigate the simpler setting where $m$ is much smaller than $n$, to demonstrate the use of the exponential potential function. This implies that in each round, deterministically at least $n - m$ bins are empty. As we prove below, this implies for example, that for $m =  \frac{n}{\log n}$ we get \Whp~a maximum load of $\Oh(\frac{\log n}{\log \log n})$ after $\Oh(\frac{n}{\log n})$ rounds.

\begin{lem}
Consider the RBB process with  $m \leq \frac{1}{e^2} n$. Then for any round $t \geq 2m$,
\[
\Pro{\max_{i \in [n]} x_i^t \leq 4\cdot \frac{\log n}{\log\left( \frac{n}{em} \right)} } \geq 1 - n^{-2}.
\]
\end{lem}
\begin{proof}
We will use the potential $\Phi := \Phi(\alpha)$ with $\alpha := \log\left(\frac{n}{em}\right) \geq 1$ (since $m \leq \frac{1}{e^2} n$). Note that for $m = o(n)$ the potential is super-exponential, as in~\cite{LS21}. Since $\kappa^t \leq m$, we have
\[
\frac{e^{\alpha} - 1}{n} \cdot \kappa^t \leq e^{\alpha} \cdot \frac{m}{n} = \frac{1}{e}. 
\]
Hence, by using \cref{lem:exp_change_general},
\begin{align*}
\Ex{\Phi^{t+1} \,\left|\, \mathfrak{F}^t \right.} 
 & \leq \Phi^t \cdot e^{-\alpha} \cdot e^{\frac{e^{\alpha} - 1}{n} \cdot \kappa^t} + (n - \kappa^t) \cdot e^{\frac{e^{\alpha} - 1}{n} \cdot \kappa^t} \\
& \leq \Phi^t \cdot e^{-\alpha} \cdot e^{1/e} + e \cdot n \\
& \leq \Phi^t \cdot e^{-\frac{\alpha}{2}} + e \cdot n, 
\end{align*}
using in the last inequality that $\alpha \geq 1$.

At round $t = 0$ we have $\Phi^0 \leq e^{\alpha m}$. Hence applying \cref{lem:geometric_arithmetic}, we have for any $t \geq 2m$,
\[
\Ex{\Phi^t} \leq e^{\alpha m} \cdot e^{- \frac{1}{2} \cdot \alpha t} + \frac{e \cdot n}{1 - e^{-\frac{\alpha}{2}}} \leq 1 + \frac{e \cdot n}{1 - e^{-1/2}} \leq 3e \cdot n.
\]
By applying Markov's inequality for any $t \geq 2m$,
\[
\Pro{\Phi^t \leq 3e \cdot n^{3}} \geq 1 - n^{-2}.
\]
When $\{ \Phi^t \leq 3e \cdot n^{3} \}$ holds, we have for any bin $i \in [n]$,
\[
x_i^t \leq \frac{1}{\alpha} \cdot \left( \log(3e) + 3\log n\right) \leq  4\cdot \frac{\log n}{\log\left( \frac{n}{em} \right)},
\]
completing the proof.
\end{proof}

\subsection{Lower Bounding the Number of Empty Bins for \texorpdfstring{$m \geq n$}{m >= n}} \label{sec:number_of_empty_bins_lb}

We now proceed to the more natural and challenging case where $m \geq n$. %

Central to our analysis is an invariant about the number of empty bins. The basic idea is inspired by~\cite[Lemma~19]{BCNPP19}, who proved that in case of $m=n$, for each round, a constant fraction of the bins are empty with very high probability. This is useful, as it implies a constant additive drift for the load of each non-empty bin, which will drop by a constant term $>0$ in expectation.

However, for general $m \gg n$, there will be starting configurations, in which all bins remain non-empty for several rounds. Only if the process runs for a sufficiently long time, a small fraction of bins will become (and, to some extent, remain) empty. The following lemma quantifies this behavior and proves that, after a waiting time of $\Oh((m/n)^2)$ (the square of the average load), a fraction of $\Oh(n/m)$ of the bins will be empty per round on average. Hence for a time interval of length $(m/n)^2$, the aggregated ``empty bin/round pairs'' will be $\approx (m/n)^2 \cdot n \cdot (n/m) =m$.

\begin{lem} \label{lem:many_empty}
Consider the RBB process with $m$ balls, where $m \geq n$ and any round $t_0 \geq 0$. Then, for round $t_3 := t_0 + 744 (m/n)^2$ it holds that
\begin{align*}
     \Pro{ \left. F_{t_0}^{t_3} \geq \frac{1}{384} \cdot m ~\right|~ \mathfrak{F}^{t_0} } \geq 1-e^{-\Omega(n)}.
\end{align*}
\end{lem}
First, let us remark that for the simpler case $m=n$, a stronger result was shown in~\cite[Lemma~1]{BCNPP19}, proving that for \emph{any} round $t \geq 1$, $F^{t}=\Omega(n)$ holds with probability $1-\exp(-\Omega(n))$. In fact adjusting the proof in~\cite{BCNPP19} slightly, the same result holds for any $m=\Oh(n)$. 
Therefore, we may assume in the following proof for convenience, that $m \geq C \cdot n$ for a sufficiently large constant $C>0$ (we will choose $C:=6$). Alternatively, we can also reduce the case with $m$ balls for some $m \in [n, C \cdot n]$ balls to the case with $C \cdot n$ balls, by using the fact that $F_{t_0}^{t_3}$ becomes stochastically smaller if we add more balls.

In order to establish \cref{lem:many_empty}, we will relate the RBB process to a simpler process, which we call the \emph{idealized process}. In the idealized process, we also remove one ball from each non-empty bin at each round, but we allocate exactly $n$ balls, regardless of how many bins are empty. 

Formally, fix any load configuration of $m$ balls with load vector $x^{t_0}$. The load vector of the idealized process is denoted by $y^{t}, t \geq t_0$ and defined as follows. For any bin $i \in [n]$, $y_i^{t_0} := x_i^{t_0}$. Further, for any $t \geq t_0$, let $Z_1^{t},Z_2^{t},\ldots,Z_n^{t} \in \{1,\ldots,n\}$ be $n$ independent, uniform random samples. Then define,
\begin{align}
  y_i^{t+1} := y_i^{t} - \mathbf{1}_{y_i^{t} > 0} + \sum_{j=1}^{n} \mathbf{1}_{Z_j^{t}=i}. \label{eq:idealized}
\end{align}
Note that the marginal distribution of $y_i^{t+1}$ can be expressed as
\[
 y_i^{t+1} = y_i^{t} - \mathbf{1}_{y_i^{t} > 0} + \Bin(n,1/n) .
\]
Comparing this to the RBB process (see~\cref{eq:marginal}) we have the same distribution apart from that $\Bin(n,1/n)$ is replaced by $\Bin(\kappa^{t},1/n)$. Thus we see that the \emph{idealized process} is a bit simpler and also has the advantage that the number of balls that are added to the bins does not depend on the load configuration.

\begin{lem} \label{lem:coupling}
For any round $t_0 \geq 0$ and load vector $x^{t_0}$, there is a coupling between the load vectors $(x^{t})_{t \geq t_0}$ and $(y^{t})_{t \geq t_0}$ such that for all rounds $t \geq t_0$ and for all $i \in [n]$,
$
 x_i^t \leq y_i^t.
$
\end{lem}
\begin{proof}
This claim follows by induction. First, the claim holds for $t=t_0$ by definition, as $y^{t_0}=x^{t_0}$. For the induction step, let $y^{t}$ and $x^{t}$ such that $ x_i^t \leq y_i^t$ for all $i \in [n]$. We define a coupling between the RBB and the idealized process such that they share the same sequence of samples $(Z_j^t)_{j \in [n]}$. Then, recall that the idealized process has
\[
 y_i^{t+1} := y_i^{t} - \mathbf{1}_{y_i^t > 0} + \sum_{j=1}^{n} \mathbf{1}_{Z_j^t=i},
\]
and we define the RBB to have
\[
 x_i^{t+1} := x_i^{t} - \mathbf{1}_{x_i^t > 0} + \sum_{j=1}^{\kappa^t} \mathbf{1}_{Z_j^t=i}.
\]
Hence, it follows that the number of balls that are added to $y_i^t$ is at least as large as the number of balls that are added to $x_i^t$. If $x_i^{t}=y_i^{t}$, then this implies that $x_i^{t+1} \leq y_i^{t+1}$. Further, if $x_i^{t} +1 \leq y_i^{t}$, then since $\mathbf{1}_{x_i^t > 0} -\mathbf{1}_{y_i^t > 0} \geq -1$, we also have $x_i^{t+1} \leq y_i^{t+1}$. This completes the induction and the lemma follows.
\end{proof}

Based on the coupling, we also define for two rounds $t_0 \leq t_3$,
\[
 G_{t_0}^{t_3} := \sum_{t=t_0}^{t_3} \sum_{i \in [n]} \mathbf{1}_{y_i^t=0}.
\]
Note that \cref{lem:coupling} implies that $F_{t_0}^{t_3}$ is stochastically larger than $G_{t_0}^{t_3}$, therefore it suffices to analyze $G_{t_0}^{t_3}$ in the following.

Our first lemma proves that starting from any load configuration with $m$ balls at time $t_0$, any bin $i \in [n]$ whose load is about the average load, has a constant probability $>0$ of reaching zero load after $\Oh(\frac{m^2}{n^2})$ rounds.
\begin{lem} \label{lem:idealized_one}
Consider the idealized process with an arbitrary initial load configuration at time $t_0$ with $m \geq 6n$ balls. Let $i \in [n]$ be any bin with $y_i^{t_0} \leq 2 \cdot m/n$. Then,
\[
 \Pro{ \left. \bigcup_{t_1 \in [t_0,t_0 + 720 \cdot \frac{m^2}{n^2}]} \left\{ y_i^{t_1} = 0 \right\} ~\right|~ \mathfrak{F}^{t_0}, y_i^{t_0} \leq 2 \cdot \frac{m}{n} } \geq \frac{1}{4}. 
\]
\end{lem}
\begin{proof}
Fix any bin $i \in [n]$ with load $y_i^{t_0} \leq 2 \cdot m/n$. Let us define the stopping time
\[
\tau := \min\left\{ t \geq t_0 \colon |y_i^{t} - y_i^{t_0}| \geq y_i^{t_0} \right\},
\]
so $\tau$ is the first time after round $t_0$ when the load if $i$ has either zeroed or doubled. Since
\[
 y_i^{t} = y_i^{t-1}  -  \mathbf{1}_{y_i^{t-1} > 0} + \Bin(n,1/n),
\]
it follows by taking expectations that for any $t \in [t_0+1,\tau]$, 
\[
 \Ex{ \left. y_i^{t} - y_i^{t-1} \, \right| \, \mathfrak{F}^{t-1} } = 0.
\]
Hence a drift argument~(\cref{lem:drift_prob}) implies that
\[
 \Pro{ y^{\tau} = 0} \geq \frac{1}{2}.
\]
In order to upper bound $\Ex{\tau}$, we define $\tilde{y}_i^{t}:=\min\{ y_i^{t}, 2 \cdot y_i^{t_0} \}$ for any $t \geq t_0$. 
Note that for any $t \in [t_0+1, \tau]$ with $\tilde{y}_i^{t-1} \geq 1$,
\[
\Pro{\tilde{y}_i^{t} = \tilde{y}_i^{t-1} - 1} \geq \Pro{\Bin(n, 1/n) = 0} = \Big( 1 - \frac{1}{n} \Big)^n \geq  e^{-2n/n} = e^{-2},
\]
using that $1 - z \geq e^{-2z}$ for any $z \in [0, 0.75]$.
Similarly, if $\tilde{y}_i^{t-1} = 0$,
\[
\Pro{\tilde{y}_i^{t} \geq 1} = \Pro{\Bin(n, 1/n) \neq 0} = 1 - \Big( 1 - \frac{1}{n} \Big)^n \geq 1 -\frac{1}{e}.
\]

This implies that $D^{t}:=\tilde{y}_i^{t} - \tilde{y}_i^{t-1}$ satisfies for any 
$t\in [t_0+1,\tau]$,
\[
 \Ex{ \left. (D^{t})^2 \, \right| \, \mathfrak{F}^{t-1} } \geq 1^2 \cdot e^{-2} > 0.
\]
Also, $D^{t} \leq 0$, as $\Ex{ \left.\tilde{y}_i^t - \tilde{y}_i^{t-1} \, \right| \, \mathfrak{F}^{t-1} } \leq \Ex{ \left. y_i^{t} - y_i^{t-1} \, \right| \, \mathfrak{F}^{t-1} }\leq 0$ for any $t \geq t_0$.
Recall $\tau = \min \{t \geq t_0: \tilde{y}_i^t =0 \vee \tilde{y}_i^t = 2 \cdot \tilde{y}_i^{t_0} \}$.
Hence by \cref{lem:drift_walk} with $s=x_i^{t_0}$ and $\sigma^2 = e^{-2} > 1/9$, it follows that 
\[
 \Ex{\tau-t_0} \leq \frac{5s^2}{\sigma^2} \leq 45 s^2.
\]
Using Markov's inequality, 
\[
\Pro{ \tau -t_0 \geq  180 s^2 } \leq \frac{1}{4},
\]
Hence, using the union bound and $s= x_i^{t_0} \leq 2 (m/n)$, we conclude that
\[
\Pro{ \bigcup_{t_1 \in [t_0,t_0 + 720 \cdot \frac{m^2}{n^2}]} \left\{ y_i^{t_1} = 0 \right\} }
 \geq \Pro{ \left\{ y_{i}^{\tau} = 0 \right\}  \cap \left\{\tau \leq t_0 + 720 \cdot \frac{m^2}{n^2} \right\} } \geq \frac{1}{2} - \frac{1}{4} = \frac{1}{4}. %
 \qedhere
\]
\end{proof}

In the next lemma, we will prove that once $y_i^{t}=0$ occurs, then with constant probability bin $i$ will have zero load in $\Omega(m/n)$ further rounds until time $\Oh((m/n)^2)$.
\begin{lem} \label{lem:idealized_two}
Consider the idealized process with an arbitrary load configuration at round $t_1$ with $m \geq 6n$ balls, such that there is a bin $i \in [n]$ with $y_i^{t_1}=0$. Then, for round $t_2:=t_1 + 24 \cdot (m/n)^2$,
\begin{align*}
    \Pro{ \sum_{t=t_1}^{t_2} \mathbf{1}_{y_i^t=0} \geq \frac{1}{6} \cdot \frac{m}{n} ~\Bigg|~ \mathfrak{F}^{t_1}, y_i^{t_1} = 0 } \geq \frac{1}{4}.
\end{align*}
\end{lem}
\begin{proof}

First, let us consider any round $s \geq t_1$ such that $y_i^s = 0$. We define
\[
 \tau(s) := \min \{t > s \colon y_i^t = 0 \wedge y_i^{t} \geq m/n \}.
\]
We seek to prove that
\[
 \Pro{ y_i^{\tau(s)} = 0 } \geq 1 - \Oh(n/m). 
\]
By \cref{lem:drift_prob}, for any $\gamma > 0$,
\begin{align*}
 \Pro{ y_i^{\tau(s)} \geq m/n ~\Big|~ y_i^{s+1} = \gamma  } \leq \frac{n}{m} \cdot \gamma.
\end{align*}
Further, as $y_i^{s}=0$, $\Pro{ y_i^{s+1} = \gamma \, \mid \, y_i^{s}=0} \leq \Pro{ \Bin (n,1/n) = \gamma} \leq 2^{-\gamma}$ (\cref{lem:binomial_bound}). Also by Markov's inequality, $\Pro{ \Bin (n,1/n) \geq m/n} \leq n/m$. Combining the last three inequalities, %
\begin{align*}
 \lefteqn{ \Pro{ y_i^{\tau(s)} \geq m/n } 
   }\\
&=\sum_{\gamma=1}^{m/n-1} \Pro{ \left. y_i^{s+1} = \gamma \, \right| \, y_i^{s} = 0} \cdot
 \Pro{ \left. y_i^{\tau(s)} \geq m/n \, \right| \, y_i^{s+1} = \gamma} + \Pro{ y_i^{s+1} \geq m/n} \\
 & \leq \sum_{\gamma=1}^{n} 2^{-\gamma} \cdot \gamma \cdot \frac{n}{m} + \frac{n}{m} \\ &\leq  3 \cdot \frac{n}{m},
\end{align*}
where the last inequality used the fact that $\sum_{\gamma=1}^{\infty} 2^{-\gamma} \cdot \gamma \leq 2$.
Therefore, 
\[
\Pro{ y_i^{\tau(s)} = 0} \geq 1 - 3 \cdot \frac{n}{m}.
\]
Next define $\rho :=\min\{t \geq t_1 \colon y_i^{t} \geq m/n \}$. With that, the probability that the load $0$ state is visited at least $\frac{1}{6} \cdot \frac{m}{n}$ times in the time-interval $[t_1,\rho]$, is at least
\begin{align}
\Pro{ \Bigl| \left\{ t_1 \leq s \leq \rho \colon y_i^{s} =0 \right\} \Bigr| \geq \frac{1}{6} \cdot \frac{m}{n}}  & \geq \left( 1 - 3 \cdot \frac{n}{m} \right)^{ \frac{1}{6} \cdot \frac{m}{n} } \geq
\frac{1}{2}.  \label{eq:step_2}
\end{align}
where the last inequality used Bernoulli's inequality, i.e., the fact that $(1+z)^{r} \geq 1 + r \cdot z$ for all $z \geq -1, r \geq 1$ (which applies since $m \geq 6n$).

The remaining part of the proof is to show that with some constant probability $>0$, $\rho=\Oh( (m/n)^2)$. 
First recall that $y_i^{t}$ has the following distribution,
\[
 y_i^{t} = y_i^{t-1} - \mathbf{1}_{y_i^{t-1} > 0} + \Bin(n,1/n).
\]
For any $t \in [t_1,\rho]$,
\begin{align*}
 \Ex{ \left. ( y_i^{t} - y_i^{t-1} )^2 \, \right| \, \mathfrak{F}^{t-1} } & = \Var{ \left. \Bin(n,1/n) - \mathbf{1}_{y_i^{t-1} > 0} \, \right|\, \mathfrak{F}^{t-1} } \\
 & = \Var{ \Bin(n,1/n) } \\
 & = n \cdot \frac{1}{n} \cdot \Big(1- \frac{1}{n} \Big) \geq \frac{1}{2}.
\end{align*}
Further for any $t \in [t_1,\rho]$, we have
\begin{align*}
 \Ex{  y_i^{t} - y_i^{t-1}  \, \left| \, \mathfrak{F}^{t-1} \right.} \geq 0.
 \end{align*}
Hence it follows by \cref{lem:drift_walk_two} that,
\begin{align}
 \Ex{ \rho - t_1 } \leq 2 \cdot \Ex{(y_i^{\rho})^2}. \label{eq:drift_before_markov}
\end{align}
We will now upper bound $\Ex{ (y_i^{\rho})^2}$.
In each round, we will apply the principle of deferred decision to expose $Z_j^t$ by looking at each of the $n$ trials individually. Thus, as soon as we know that $\sum_{\ell \in [k]} \mathbf{1}_{Z_\ell^t = i} \geq (m/n-y_i^{t-1})+1$ for some $t$ and $k$, we know that $y_i^{t} \geq m/n$ and hence $\rho=t$. Conditional on $\sum_{\ell \in [k]} \mathbf{1}_{Z_\ell^t = i} \geq (m/n-y_i^{t-1})+1$ means that, in distribution,
\[
 y_i^{t} = y_i^{\rho} = m/n + \Bin( n - k, 1/n),
\]
for some $n \geq k \geq m/n - y_i^{\rho-1}$. Clearly, the random variable $\Bin( n - k, 1/n)$ is stochastically the largest if $k$ is as small as possible. Thus we pessimistically take $k=0$, and obtain
\begin{align*}
 \Ex{(y_i^{\rho})^2} &\leq \Ex{ ( m/n + \Bin(n,1/n) )^2} \\
 &\leq 2 \cdot (m/n)^2 + 2 \cdot \Ex{ (\Bin(n,1/n))^2 } \\
 &\leq 3 \cdot (m/n)^2,
\end{align*}
using that $m \geq 6 \cdot n$.

Using this in \cref{eq:drift_before_markov}, and applying Markov's inequality,
\begin{align}
    \Pro{ \rho - t_1 \geq 24 \cdot (m/n)^2 } \leq \frac{1}{4}. \label{eq:step_3}
\end{align}

We can now conclude the argument by a union bound. 
First, we proved in \eqref{eq:step_2} that with probability at least $\frac{1}{2}$, the load $0$ state will be visited at least $\frac{1}{6} \cdot \frac{m}{n}$ times before reaching a load level larger than $\frac{m}{n}$. Secondly, the event in \eqref{eq:step_3} ensures that a load level larger than $\frac{m}{n}$ will be reached before round $t_1+24\cdot (m/n)^2$. Therefore, with $t_2=t_1 + 24 \cdot (m/n)^2$,
\[
    \Pro{ \sum_{t=t_0}^{t_2} \mathbf{1}_{x_i^{t}=0} \geq \frac{1}{6} \cdot \frac{m}{n}} \geq \frac{1}{2} - \frac{1}{4} = \frac{1}{4}. \qedhere
\]
\end{proof}

By combining \cref{lem:idealized_one} and \cref{lem:idealized_two}, we can easily derive the following lower bound on $\Ex{G_{t_0}^{t_3}}$ (and so also on $\Ex{ F_{t_0}^{t_3} }$).

\begin{lem} \label{lem:idealized_three}
Consider the idealized process with $m$ balls, where $m \geq 6 \cdot n$ and any round $t_0 \geq 0$. Then for round $t_3 := t_0 + 744 (m/n)^2$ it holds that
\begin{align*}
     \Ex{ \left. G_{t_0}^{t_3} \, \right| \, \mathfrak{F}^{t_0} } \geq \frac{1}{192} \cdot m.
\end{align*}
\end{lem}
\begin{proof}
Decomposing, we obtain
\begin{align*}
 \Ex{  G_{t_0}^{t_3} \, \mid \, \mathfrak{F}^{t_0} } 
 &\geq \sum_{i \in [n] \colon y_i^{t_0} \leq 2 \cdot m/n }  \Ex{ \sum_{t=t_0}^{t_3} \mathbf{1}_{y_i^{t} = 0} ~\Bigg|~ \mathfrak{F}^{t_0} }.
\end{align*}
In order to lower bound the last expectation for some bin $i \in [n]$ with $y_i^{t_0} \leq 2 \cdot m/n$, let us define $t_1$ as a stopping time, i.e., $t_1:= \min \{t \colon r \geq t_0 \colon y_{i}^{r} = 0 \}$,
\begin{align*}
\lefteqn{ \Pro{\left. \sum_{t=t_0}^{t_3} \mathbf{1}_{y_i^t=0} \geq \frac{1}{6} \cdot \frac{m}{n} ~\right|~ \mathfrak{F}^{t_0}} } \\
&\geq 
\Pro{\left. \bigcup_{r \in [t_0,t_0 + 720 \cdot \frac{m^2}{n^2}]} \left( \left\{ t_1 = r \right\} \cap \left\{ \sum_{t=t_1}^{t_3} \mathbf{1}_{y_i^t=0} \geq \frac{1}{6} \cdot \frac{m}{n} \right\}  \right) ~\right|~ \mathfrak{F}^{t_0}} \\ 
&= \sum_{r=t_0}^{t_0+ 720 \cdot \frac{m^2}{n^2}} \Pro{\left. \left\{ t_1 = r \right\} \cap \left\{ \sum_{t=t_1}^{t_3} \mathbf{1}_{y_i^t=0} \geq \frac{1}{6} \cdot \frac{m}{n} \right\} ~\right|~ \mathfrak{F}^{t_0}
} \\
&\stackrel{(a)}{\geq}  \sum_{r=t_0}^{t_0+ 720 \cdot \frac{m^2}{n^2}} \Pro{ t_1 = r ~\left|~ \mathfrak{F}^{t_0} \right.} \cdot
\Pro{ \left. \sum_{t=t_1}^{t_2} \mathbf{1}_{y_i^t=0} \geq \frac{1}{6} \cdot \frac{m}{n} ~\right|~ \mathfrak{F}^{t_1}, t_1 = r, y_i^{r} = 0} \\
&\stackrel{(b)}{\geq} \frac{1}{4} \cdot \Pro{ t_1 \in \left[t_0,t_0+720 \cdot \frac{m^2}{n^2} \right] ~\Bigg|~ \mathfrak{F}^{t_0} } \\
&\stackrel{(c)}{\geq} \frac{1}{4} \cdot \frac{1}{4}.
\end{align*}
where in $(a)$ we used the definition of $t_2:=t_1+ 24 \cdot (m/n)^2 $ (see~\cref{lem:idealized_two}), $(b)$ we used \cref{lem:idealized_three} and in $(c)$ we used \cref{lem:idealized_two}. This implies for the expectation for any bin $i$ with $y_i^{t} \leq 2 \cdot (m/n)$,
\[
\Ex{ \left. \sum_{t=t_0}^{t_3} \mathbf{1}_{y_i^{t} = 0} ~\right|~ \mathfrak{F}^{t_0} } \geq \frac{1}{4} \cdot \frac{1}{4} \cdot \frac{1}{6} \cdot \frac{m}{n} = \frac{1}{96} \cdot \frac{m}{n}.
\]
Since at any time, in particular at time $t_0$, it holds deterministically that at least half of all bins have load at most $2 \cdot m/n$, it follows that
\[
 \Ex{ \left. G_{t_0}^{t_3} \, \right| \, \mathfrak{F}^{t_0} } \geq \frac{1}{192} \cdot \frac{m}{n},
\]
which completes the proof.
\end{proof}

\begin{lem}\label{lem:lipschitz}
Consider the idealized process with $m$ balls, where $m \geq 6 \cdot n$.
For any $\mathfrak{F}^{t_0}$, define $f:=f((Z_{i}^{t})_{t \in [t_0,t_3),i \in [n]}):=G_{t_0}^{t_3}$, where $t_0 \leq t_3$ are arbitrary rounds. This is a function of the set of independent random variables $Z_i^{t}, i \in [n], t \in [t_0,t_3]$ (see \cref{eq:idealized}). Then changing one $Z_j^{t}$ changes $G_{t_0}^{t_3}$ by at most $1$.
\end{lem}
\begin{proof}
Fix an arbitrary load vector $y^{t_0}$. By definition of the idealized process (see, e.g.,~\cref{eq:idealized}), for any $i \in [n]$ and $t \in [t_0,t_3)$,
\[
  y_i^{t+1} := y_i^{t} - \mathbf{1}_{y_i^{t} > 0} + \sum_{j=1}^{n} \mathbf{1}_{Z_j^{t}=i}.
\]
Note that any assignment of bins $(z_i^t)_{t \in [t_0,t_3), i \in [n]}$ to the random variables $(Z_i^{t})_{t \in [t_0,t_3),i \in [n]}$ such that $Z_i^t = z_i^t$ completely determines the evolution of all $y^{t}$ for $t \in [t_0,t_3]$.

For any $j \in [n]$ and any $r \in [t_0,t_3]$, we now consider an alternative assignment of bins $\tilde{z}_i^{t}$ to coincide with $z_{i}^{t}$ apart from $z_{j}^{r}$, i.e.,
\begin{equation*}
    \tilde{z}_i^{t} ~
    \begin{cases}
     = z_i^t & \mbox{ if $i \neq j$ or $t \neq r$, } \\
     \neq z_i^{t} & \mbox{ if $i=j$ and $t=r$.} \\
    \end{cases}
\end{equation*}
To simplify notation, let $k:=z_j^r$ and $\tilde{k}:=\tilde{z}_j^r$; recall that $k \neq \tilde{k}$. Also let us denote by $\tilde{y}^{t}$ the load vector determined by the assignment $\tilde{z}$. Further, for any $\ell \in [n]$, let us denote by $e(\ell)$ the $n$-dimensional unit-vector defined by $e(\ell)_{i}:=\mathbf{1}_{\ell=i}$, $i \in [n]$. Then it is clear that $\tilde{y}^{t} = y^{t}$ for any $t \in [t_0,r)$, and 
\[
\tilde{y}^{r} + e(k) = y^{r} + e(\tilde{k}).
\]
which means that the load vectors $\tilde{y}$ and $y$ have a $\ell_1$-distance of $2$, since $k \neq \tilde{k}$. Since $\tilde{z}_i^t = z_i^t$ for any $t \in (r,t_3)$, it follows that both $\tilde{y}$ and $y$ receive the same number of balls in each round, i.e., for any $t \in (r,t_3)$,
\[
 \sum_{j=1}^{n} \mathbf{1}_{z_j^t = i} = \sum_{j=1}^{n} \mathbf{1}_{\tilde{z}_j^t = i}.
\]
The only difference may come from the terms $-\mathbf{1}_{\tilde{y}_i^t >0}$ and $-\mathbf{1}_{y_i^t >0}$. 
To this end, let us define two stopping times $\rho(k)$ and $\rho(\tilde{k})$, when the bins $k$ and $\tilde{k}$, respectively, become empty,
\begin{align*}
    \rho(k) := \min\left\{ t \geq r \colon \tilde{y}_k^t = 0 \right\}, \\
    \rho(\tilde{k}) := \min\left\{t \geq r \colon y_{\tilde{k}}^t = 0 \right\}.
\end{align*}
Note that by simple induction, for any $t \in (r,\rho(k))$,
\[
  \tilde{y}_{k}^{t} = y_{k}^{t} + 1.
\]
At iteration $\rho(k)$, $\tilde{y}_k^{t}=0$ but $y_{k}^{t}=1$, so $\mathbf{1}_{ \tilde{y}_k^{t} > 0}=0$ and $\mathbf{1}_{ y_k^{t} > 0}=1$ and thus
\[
 \tilde{y}_{k}^{\rho(k)+1} = y_{k}^{\rho(k)+1},
\]
and more generally, for any $t \geq \rho(k)+1$,
\[
\tilde{y}_{k}^{t} = y_{k}^{t}.
\]
The analogous argument holds for the stopping time $\rho(\tilde{k})$ and bin $\tilde{k}$. Thus it follows that we compare $G_{t_0}^{t_3}$ (which is defined via the load vector $y$) and $\tilde{G}_{t_0}^{t_3}$ (which is defined via the load vector $\tilde{y}$), the only difference occurs at rounds $\rho(k)$ when bin $k$ is empty in $\tilde{y}$ but not $y$, and at round $\rho(\tilde{k})$ when bin $\tilde{k}$ is empty in $y$ but not in $\tilde{y}$. Hence $G_{t_0}^{t_3}$ and $\tilde{G}_{t_0}^{t_3}$ differ by exactly one if $\rho(k) < t_3$ and $\rho(\tilde{k}) > t_3$ (or vice versa), while in all other cases, they are equal.
\end{proof}

\begin{lem} \label{lem:idealized_concentration}
Consider the idealized process with $m$ balls, where $m \geq 6 \cdot n$. Then, for any $t_0 \geq 0$ and $t_3=t_0+744 \cdot (m/n)^2$,
\[
 \Pro{ \left. G_{t_0}^{t_3} \geq  \frac{1}{384} \cdot m ~\right|~ \mathfrak{F}^{t_0} } \geq 1 - e^{-\Omega(n)}.
\]
\end{lem}
\begin{proof}
By \cref{lem:lipschitz}, $f:=f((Z_{i}^{t})_{t \in [t_0,t_3],i \in [n]})=G_{t_0}^{t_3}$, $f$ is a function of $(t_3-t_0+1) \cdot n$ independent random variables. By \cref{lem:idealized_three},
\[
 \Ex{ f ~\Big|~ \mathfrak{F}^{t_0} } \geq \frac{1}{192} \cdot m.
\]
As shown in \cref{lem:lipschitz}, changing one $Z_i^t$ changes $f$ by at most $1$.
Hence by the Method of Bounded Differences (\cref{thm:mobd}),
\begin{align*}
     \Pro{ G_{t_0}^{t_3} \leq \Ex{ G_{t_0}^{t_3}} - \lambda ~\Big|~ \mathfrak{F}^{t_0} } \leq \exp\left(- \frac{\lambda^2}{2 \sum_{t=t_0}^{t_3} \sum_{i=1}^{n} 1^2 } \right).
\end{align*}
Choosing $\lambda = \frac{1}{384} \cdot m$ yields
\[
  \Pro{ \left. G_{t_0}^{t_3} \leq \frac{1}{384} \cdot m ~\right|~ \mathfrak{F}^{t_0} } \leq 
     \exp\left(- \frac{ 2 \left( \frac{1}{384} \cdot m \right)^2 }{ 744 (m/n)^2 \cdot n} \right) \leq \exp(-\Omega(n)).\qedhere
\]
\end{proof}
As mentioned before, \cref{lem:coupling} implies that $F_{t_0}^{t_3}$ is stochastically larger than $G_{t_0}^{t_3}$, hence \cref{lem:many_empty} follows immediately from \cref{lem:idealized_concentration}.

\subsection{Bounding the Convergence Time for \texorpdfstring{$m \geq n$}{m >= n}} \label{sec:convergence_time}

We will now bound the convergence time for $m \geq n$. In particular, we will analyze the exponential potential function for $\alpha = \Theta(n/m)$ and show that in $\Oh(m^2/n)$ rounds the process reaches a configuration with $\Phi^t < \frac{48}{\alpha^2} \cdot n$, which implies a maximum load of $\Oh(\alpha^{-1} \cdot (\log(\alpha^{-1}) + \log n)) = \Oh(m/n \cdot \log m)$, which is $\Oh(m/n \cdot \log n)$ for $n \leq m \leq \poly(n)$.

We start by proving that potential drops in expectation when it is sufficiently large and there is a large fraction of empty bins.
\begin{lem} \label{lem:exponential_drop}
Consider the RBB process for any $m \geq n$, and the potential $\Phi := \Phi(\alpha)$ with $\alpha := \frac{1}{2 \cdot 384 \cdot 744} \cdot \frac{n}{m}$. Then for any round $t \geq 0$,
\[
\Ex{\left. \Phi^{t+1} \, \right|\, \mathfrak{F}^t} \leq \Phi^t \cdot e^{\alpha^2 - \alpha f^t} + 6n.
\]
In particular,
\[
\Ex{\Phi^{t+1} \,\, \left| \,\, \mathfrak{F}^t, \Phi^t > \frac{48}{\alpha^2} \cdot n \right.} \leq \Phi^t \cdot e^{1.5\alpha^2 - \alpha f^t}.
\]
\end{lem}

\begin{proof}

Using \cref{lem:exp_change_general},
\begin{align*}
\Ex{\left. \Phi^{t+1} \,\right|\, \mathfrak{F}^t} 
 & \leq \Phi^t \cdot e^{-\alpha} \cdot e^{\frac{e^{\alpha} - 1}{n} \cdot \kappa^t} + (n - \kappa^t) \cdot e^{\frac{e^{\alpha} - 1}{n} \cdot \kappa^t} \\
 & \leq \Phi^t \cdot e^{-\alpha} \cdot e^{\frac{e^{\alpha} - 1}{n} \cdot \kappa^t} + 6n,
\end{align*}
since $e^{\frac{e^{\alpha} - 1}{n} \cdot \kappa^t} \leq e^{e^{\alpha} - 1} \leq 6$, as $\alpha \leq 1$ and $\kappa^t \leq n$.

When $\alpha < 1.75$, we have that $e^{\alpha} \leq 1 + \alpha + \alpha^2$, so
\[
\Ex{\left. \Phi^{t+1} \,\right|\, \mathfrak{F}^t} \leq \Phi^t \cdot e^{-\alpha} \cdot e^{(\alpha + \alpha^2) \cdot \frac{\kappa^t}{n}} + 6n,
\]
which proves the first statement.

For the second statement, in particular, since $\kappa^t = n \cdot (1 - f^t)$, we get
\[
\Ex{\left. \Phi^{t+1} \,\right|\, \mathfrak{F}^t} \leq \Phi^t \cdot e^{\alpha^2 - \alpha f^t} + 6n.
\]
Assuming that $\Phi^t > \frac{48}{\alpha^2} \cdot n$, the first term dominates the additive increase,
\begin{align*}
\Ex{\Phi^{t+1} \,\, \left| \,\, \mathfrak{F}^t, \Phi^t > \frac{48}{\alpha^2} \cdot n \right.} 
 & \leq \Phi^t \cdot e^{1.5 \alpha^2 - \alpha f^t} \cdot e^{-0.5 \alpha^2} + 6n \\
 & \stackrel{(a)}{\leq} \Phi^t \cdot e^{1.5 \alpha^2 - \alpha f^t} \cdot \Big( 1 - \frac{1}{4} \alpha^2 \Big)  + 6n \\
 & \leq \Phi^t \cdot e^{1.5 \alpha^2 - \alpha f^t} - \frac{1}{4} \alpha^2 \cdot e^{-\alpha} \cdot \frac{48}{\alpha^2} \cdot n + 6n \\
 & \stackrel{(b)}{\leq} \Phi^t \cdot e^{1.5 \alpha^2 - \alpha f^t} - 6n + 6n \\
 & = \Phi^t \cdot e^{1.5 \alpha^2 - \alpha f^t},
\end{align*}
using in $(a)$ that $e^z \leq 1 + \frac{1}{2} z $ for $z \in [-1, 0]$ and in $(b)$ that $e^{-\alpha} \geq \frac{1}{2}$.
\end{proof}

We now define the event
\[
\mathcal{E}^t := \left\lbrace \Phi^t \leq \frac{48}{\alpha^2} \cdot n \right\rbrace.
\]
When $\mathcal{E}^t$ holds, the potential is small enough to imply a maximum load of $\Oh(m/n \cdot \log m)$. When it is larger, it drops by a multiplicative factor. We now define for any $t_0 \geq 0$, the \textit{adjusted exponential potential function} $\tilde{\Phi}_{t_0} := \tilde{\Phi}_{t_0}(\alpha)$, is defined as $\tilde{\Phi}_{t_0}^{t_0} := \Phi^{t_0}(\alpha)$ and for any $s > t_0$
\begin{align*}
\tilde{\Phi}_{t_0}^s := \mathbf{1}_{\cap_{t \in [t_0, s)} \neg \mathcal{E}^t} \cdot \Phi^s(\alpha) \cdot e^{\sum_{t = t_0}^{s-1} (\alpha f^t - 1.5 \alpha^2)}.
\end{align*}
We will now show that it forms a super-martingale.

\begin{restatable}{lem}{LemPhiSupermartingale}\label{lem:phi_supermartingale}
For any $t_0 \geq 0$ and $\alpha > 0$ as defined in \cref{lem:exponential_drop}, the sequence $(\tilde{\Phi}_{t_0}^s(\alpha))_{s \geq t_0}$ forms a super-martingale.
\end{restatable}
\begin{proof}
Consider any $t_0 \geq 0$. Then we need to prove for any $s \geq t_0$ that,
\begin{align}
   \Ex{ \left. \tilde{\Phi}_{t_0}^{s+1} \, \right| \, \mathfrak{F}^{s}} \leq  \tilde{\Phi}_{t_0}^{s}. \label{eq:submartingale}
\end{align}
We continue with a case distinction.

\textbf{Case A [$\mathcal{E}^{s}$ holds]:} Then $\mathbf{1}_{\cap_{t \in [t_0, s+1)} \mathcal{E}^t} = 0$ and $\tilde{\Phi}_{t_0}^{s+1} = 0$, so the inequality in \cref{eq:submartingale} holds.

\textbf{Case B [$\mathcal{E}^{s}$ does not hold]:} Then 
\begin{align*}
\Ex{\left. \tilde{\Phi}_{t_0}^{s+1} \,\right|\, \mathfrak{F}^s, \neg \mathcal{E}^s} 
 & = \Ex{\left. \mathbf{1}_{\cap_{t \in [t_0, s+1)} \neg \mathcal{E}^t} \cdot \Phi^{s+1} \cdot e^{\sum_{t = t_0}^{s} (\alpha f^t - 1.5\alpha^2)} \,\, \right|\,\,  \mathfrak{F}^s, \neg \mathcal{E}^{s}} \\
 & = \mathbf{1}_{\cap_{t \in [t_0, s)} \neg \mathcal{E}^t} \cdot e^{\sum_{t = t_0}^{s} (\alpha f^t - 1.5\alpha^2)} \cdot \Ex{\mathbf{1}_{\neg \mathcal{E}^s} \cdot \Phi^{s+1} \,\left|\, \mathfrak{F}^s, \neg \mathcal{E}^{s} \right.} \\
 & = \mathbf{1}_{\cap_{t \in [t_0, s)} \neg \mathcal{E}^t} \cdot e^{\sum_{t = t_0}^{s} (\alpha f^t - 1.5\alpha^2)} \cdot \Ex{\Phi^{s+1} \,\left|\, \mathfrak{F}^s, \neg \mathcal{E}^{s} \right.} \\
 & \leq \mathbf{1}_{\cap_{t \in [t_0, s)} \neg \mathcal{E}^t} \cdot e^{\sum_{t = t_0}^{s} (\alpha f^t - 1.5\alpha^2)} \cdot \Phi^s \cdot e^{1.5\alpha^2 - \alpha f^s} \\
 & \leq \mathbf{1}_{\cap_{t \in [t_0, s)} \neg \mathcal{E}^t} \cdot \Phi^s \cdot e^{\sum_{t = t_0}^{s-1} (\alpha f^t - 1.5\alpha^2)} \\
 & = \tilde{\Phi}_{t_0}^s,
\end{align*}
using that $\ex{\Phi^{s+1} \mid \mathfrak{F}^s, \Phi^s > \frac{48}{\alpha^2} \cdot n} \leq \Phi^s \cdot e^{1.5\alpha^2 - \alpha f^s}$ by \cref{lem:exponential_drop}.
\end{proof}

Using \cref{lem:many_empty}, we will show that in a $\Theta(m^2/n)$ interval \Whp~the potential becomes small at least once, implying an $\Oh(m/n \cdot \log m)$ bound on the maximum load.

\begin{lem}[Convergence] \label{lem:convergence}
Consider the RBB process for any $m \geq n$ and the potential $\Phi := \Phi(\alpha)$ for $\alpha > 0$ as defined in \cref{lem:exponential_drop}. Let $c_r := 16 \cdot 384^2 \cdot 744^2$. Then, for any $t_0 \geq 0$, for $t_1 := t_0 + c_r \cdot \frac{m^2}{n}$, we have
\[
\Pro{\bigcup_{t \in [t_0, t_1]} \left\{ \Phi^t \leq \frac{48}{\alpha^2} \cdot n \right\} } \geq 1 - e^{-\Omega(n)}.
\]
In particular, this implies that for $m = \poly(n)$, there exists a constant $C > 0$
\[
\Pro{\bigcup_{t \in [t_0, t_1]} \left\lbrace \max_{i \in [n]} x_i^t \leq C \cdot \frac{m}{n} \cdot \log m  \right\rbrace } \geq 1 - e^{-\Omega(n)}.
\]
\end{lem}
\begin{proof} 
Let $\Delta := 744 \cdot \frac{m^2}{n^2}$. We partition $[t_0, t_1]$ into $16 \cdot 384^2 \cdot 744 \cdot n$ intervals of length $\Delta$. Applying \cref{lem:many_empty} to each of these intervals and taking the union bound over these, we get
\begin{align*}
\Pro{ \bigcap_{i \in [0, 16 \cdot 384^2 \cdot 744 \cdot n)} \left\{ F_{t_0 + \Delta \cdot i}^{t_0 + \Delta \cdot (i+1) } \geq \frac{1}{384} \cdot m \right\} } \geq 
1-e^{-\Omega(n)}.
\end{align*}
By aggregating, this implies that
\begin{align*}
\Pro{F_{t_0}^{t_1} \geq 16 \cdot 384 \cdot 744 \cdot m \cdot n } \geq 1-e^{-\Omega(n)}.
\end{align*}
By \cref{lem:phi_supermartingale},
\[
 \Ex{\left. \tilde{\Phi}_{t_0}^{t_1} \, \right| \, \mathfrak{F}^{t_0} } \leq \tilde{\Phi}_{t_0}^{t_0} = \Phi^{t_0}.
\]
Hence by Markov's inequality we also have that
\[
\Pro{\tilde{\Phi}_{t_0}^{t_1} \leq e^{n} \cdot \Phi^{t_0}} \geq 1 - e^{-n}.
\]
Assume now that $\left\lbrace F_{t_0}^{t_1} \geq 16 \cdot 384 \cdot 744 \cdot m \cdot n \right\rbrace$ and $\left\lbrace \tilde{\Phi}_{t_0}^{t_1} \leq e^n \cdot \Phi^{t_0} \right\rbrace$ both hold. Then,
    \begin{align*}
  \tilde{\Phi}_{t_0}^{t_1} &= \mathbf{1}_{\cap_{t \in [t_0,t_1)} \neg \mathcal{E}^t } \cdot   \Phi^{t_1} \cdot e^{\sum_{t = t_0}^{t_1-1} (\alpha f^t - 1.5 \alpha^2)} \leq e^{n} \cdot \Phi^{t_0},
\end{align*}
which implies that, %
\begin{align*}
 \mathbf{1}_{\cap_{t \in [t_0,t_1)} \neg \mathcal{E}^t } \cdot \Phi^{t_1} 
 & \leq e^{n} \cdot \Phi^{t_0} \cdot e^{-\sum_{t = t_0}^{t_1-1} (\alpha f^t - 1.5 \alpha^2)} \\
 & = e^{n} \cdot \Phi^{t_0} \cdot e^{- \alpha \cdot \frac{1}{n} \cdot \sum_{t = t_0}^{t_1-1} F^t + 1.5 \cdot (t_1 - t_0) \cdot \alpha^2} \\
 & \stackrel{(a)}{\leq} e^{n} \cdot e^{\alpha m} \cdot e^{-\alpha \cdot 16 \cdot 384 \cdot 744 \cdot m + 1.5 \alpha^2 \cdot 16 \cdot 384^2 \cdot 744^2 \cdot \frac{m^2}{n}} \\
 & = e^{n} \cdot e^{\alpha m} \cdot e^{-8 \cdot n + 6 \cdot n} \\
 & \stackrel{(b)}{=} e^{n} \cdot e^{n/(2 \cdot 384 \cdot 744)} \cdot e^{-2 \cdot n} \\
 & < 1,
\end{align*}
where in $(a)$ we used that for any $t \geq 0$, $\Phi^{t} \leq e^{\alpha m}$ (the potential is maximized if all $m$ balls are in the same bin) and in $(b)$ the definition $\alpha := \frac{1}{2 \cdot 384 \cdot 744} \cdot \frac{n}{m}$. Since $\Phi^{t_1} \geq e^{\alpha \cdot 1} \geq 1$ deterministically, this implies that $\big\{ \tilde{\Phi}^{t_1} = 0 \big\}$ holds and so $\left\{ \cup_{t \in [t_0, t_1)} \mathcal{E}^t \right\}$ also holds. Hence, by the union bound,
\[
\Pro{\bigcup_{t \in [t_0, t_1]} \left\lbrace \Phi^t \leq \frac{48}{\alpha^2} \cdot n \right\rbrace} \geq 1 - e^{-n} - e^{-\Omega(n)} = 1 - e^{-\Omega(n)}.
\]
To get the bound on the maximum load, note that when the event $\left\lbrace \Phi^t \leq \frac{48}{\alpha^2} \cdot n \right\rbrace$ holds, then for any $i \in [n]$, 
\[
x_i^t \leq \frac{1}{\alpha} \cdot \log \Phi^t \leq \frac{1}{\alpha} \cdot \Big(\log 48 -2 \log \alpha + \log n\Big) \leq \kappa \cdot \frac{m}{n} \cdot \log m, 
\]
for some constant $\kappa > 0$, since $\alpha = \Theta(n/m)$ and $m \geq n$.
\end{proof}

\subsection{Maximum Load Upper Bound for \texorpdfstring{$n \leq  m \leq \poly(n)$}{n <= m <= poly(n)}} \label{sec:stabilization}

We will now show that, for any $n \leq m \leq \poly(n)$, once a configuration with $\Phi^t \leq \frac{48}{\alpha^2} \cdot n$ is reached, then \Whp~the process will re-visit such a configuration in the next $\Oh(m^2/n\cdot \log n)$ rounds. The proof is quite similar to \cref{lem:convergence}, but with intervals of shorter lengths. By a \OneChoice argument we will deduce that the maximum load in every of the in-between rounds is $\Oh(m/n \cdot \log n)$ and so the maximum load remains small for $\poly(n)$ rounds.

\begin{lem} \label{lem:returns_to_stable}
Consider the RBB process with $n \leq m \leq n^{k}$ for some constant $k \geq 1$  and the potential $\Phi := \Phi(\alpha)$ for $\alpha > 0$ as defined in \cref{lem:exponential_drop}. Further, let $c_s := 8k \cdot 16 \cdot 384^2 \cdot 744^2$. Then, for any round $t_0 \geq 0$ and for $t_1 := t_0 + c_s \cdot \frac{m^2}{n^2} \cdot \log n$, we have
\[
\Pro{\bigcup_{t \in [t_0, t_1]} \left\{ \Phi^t \leq \frac{48}{\alpha^2} \cdot n \right\} \, ~\Bigg|~ \, \mathfrak{F}^{t_0}, \Phi^{t_0} \leq e^{\alpha \log n} \cdot \frac{48}{\alpha^2} \cdot n } \geq 1 - n^{-7k}.
\]
\end{lem}
\begin{proof}
Similarly to \cref{lem:convergence}, let $\Delta := 744 \cdot \frac{m^2}{n^2}$. By applying \cref{lem:many_empty} and a union bound over the $8k \cdot 16 \cdot 384^2 \cdot 744 \cdot \log n$ intervals of length $\Delta$, we have %
\begin{align*}
\Pro{ \left. \bigcap_{i \in [0, 8k \cdot 16 \cdot 384^2 \cdot 744 \cdot \log n)} \left\{ F_{t_0 + \Delta \cdot i}^{t_0 + \Delta \cdot (i+1) } \geq \frac{1}{384} \cdot m \right\} ~ \right\vert ~ \mathfrak{F}^{t_0}, \Phi^{t_0} \leq e^{\alpha \log n} \cdot \frac{48}{\alpha^2} \cdot n } \geq 
1-e^{-\Omega(n)},
\end{align*}
which, by aggregating, implies that,
\begin{align*}
\Pro{F_{t_0}^{t_1} \geq 8k \cdot 16 \cdot 384 \cdot 744 \cdot m \cdot \log n ~ \left\vert ~ \mathfrak{F}^{t_0}, \Phi^{t_0} \leq e^{\alpha \log n} \cdot \frac{48}{\alpha^2} \cdot n \right. } \geq 1-e^{-\Omega(n)}.
\end{align*}
By \cref{lem:phi_supermartingale},
\[
 \Ex{ \tilde{\Phi}_{t_0}^{t_1} \, \left| \, \mathfrak{F}^{t_0} \right.} \leq \tilde{\Phi}_{t_0}^{t_0} = \Phi^{t_0}.
\]
Hence, by Markov's inequality we also have that
\[
\Pro{\tilde{\Phi}_{t_0}^{t_1} \leq n^{8k} \cdot \Phi^{t_0}} \geq 1 - n^{-8k}.
\]
Assume now that $\left\lbrace F_{t_0}^{t_1} \geq 8k \cdot 16 \cdot 384 \cdot 744 \cdot m \cdot \log n \right\rbrace$ and $\left\lbrace \tilde{\Phi}_{t_0}^{t_1} \leq n^{8k} \cdot \Phi^{t_0}  \right\rbrace$ both hold. Then,
\begin{align*}
  \tilde{\Phi}_{t_0}^{t_1} &= \mathbf{1}_{\cap_{t \in [t_0,t_1)} \neg \mathcal{E}^t } \cdot   \Phi^{t_1} \cdot e^{\sum_{t = t_0}^{t_1-1} (\alpha f^t - 1.5 \alpha^2)} \leq n^{8k} \cdot \Phi^{t_0}.
\end{align*}
which implies that
\begin{align*}
 \mathbf{1}_{\cap_{t \in [t_0,t_1)} \neg \mathcal{E}^t } \cdot \Phi^{t_1} 
 & \leq n^{8k} \cdot \Phi^{t_0} \cdot e^{-\sum_{t = t_0}^{t_1-1} (\alpha f^t - 1.5 \alpha^2)} \\
 & \stackrel{(a)}{\leq} n^{8k} \cdot e^{\alpha \log n} \cdot \frac{48}{\alpha^2} \cdot n \cdot e^{- \alpha \cdot \frac{1}{n} \cdot \sum_{t = t_0}^{t_1-1} F^t + 1.5 \cdot (t_1 - t_0) \cdot \alpha^2} \\
 & \stackrel{(b)}{\leq} n^{8k + 4k} \cdot e^{-\alpha \cdot 8k \cdot 16 \cdot 384 \cdot 744 \cdot \frac{m}{n} \cdot \log n + 1.5 \alpha^2 \cdot 8k \cdot 16 \cdot 384^2 \cdot 744^2 \cdot \frac{m^2}{n^2} \cdot \log n} \\
 & \stackrel{(c)}{=} n^{12k} \cdot e^{-(8 - 6) \cdot 8k \cdot \log n} \\
 & = n^{-4k}  < 1,
\end{align*}
where in $(a)$ we used that $\Phi^{t_0} \leq e^{\alpha \log n} \cdot \frac{48}{\alpha^2} \cdot n$,
in $(b)$ that $e^{\alpha \log n} \cdot \frac{48}{\alpha^2} \cdot n \leq n \cdot m^2 \cdot n \leq n^{4k}$ by the definition $\alpha := \frac{1}{2 \cdot 384 \cdot 744} \cdot \frac{n}{m}$ and in $(c)$ again the definition of $\alpha$. Since $\Phi^{t_1} \geq 1$ deterministically, it implies that $\tilde{\Phi}^{t_1} = 0$ and so $\left\lbrace \bigcup_{t \in [t_0, t_1)} \mathcal{E}^t \right\rbrace$ holds. Hence, by the union bound, 
\[
\Pro{\bigcup_{t \in [t_0, t_1]} \left\lbrace \Phi^t \leq \frac{48}{\alpha^2} \cdot n \right\rbrace \, ~\Bigg|~ \, \mathfrak{F}^{t_0}, \Phi^{t_0} \leq e^{\alpha \log n} \cdot \frac{48}{\alpha^2} \cdot n} \geq 1 - n^{-8k} - e^{-\Omega(n)} \geq 1 - n^{-7k}. \qedhere
\]
\end{proof}

Finally, combining \cref{lem:convergence}  and \cref{lem:returns_to_stable} we can derive the following upper bound on the maximum load, which holds for $\poly(n)$ rounds.

\begin{thm}[Stabilization] \label{thm:stabilization}
Consider the RBB process with any $n \leq m \leq n^k$ for some constant $k \geq 1$. There exists a constant $C > 0$ such that, for any $t \geq c_r \cdot \frac{m^2}{n}$, where $c_r > 0$ is the constant defined in \cref{lem:convergence},
\[
\Pro{\bigcap_{s \in [t, t + m^2]} \left\lbrace \max_{i \in [n]} x_i^s \leq C \cdot \frac{m}{n} \cdot \log n \right\rbrace } \geq 1 - n^{-2k}.
\]
\end{thm}

\begin{proof}
By \cref{lem:convergence} applied for the round $t-c_r \cdot \frac{m^2}{n}$,
\begin{align}
\Pro{\bigcup_{t_0 \in [t - c_r \cdot \frac{m^2}{n}, t]}  \left\lbrace \Phi^{t_0} \leq \frac{48}{\alpha^2} \cdot n \right\rbrace } \geq 1 - e^{-\Omega(n)}. \label{eq:exists_t0}
\end{align}
Consider an arbitrary round $r \geq 0$. In the next round with probability at least $1 - n^{-\omega(1)}$, each bin receives at most $\log n$ balls. So,\begin{align} \label{eq:one_step_bound}
\Pro{\Phi^{r+1} \leq e^{\alpha \log n} \cdot \Phi^r ~\left\vert ~ \mathfrak{F}^r \right.} \geq 1 - n^{-\omega(1)}.
\end{align}

We consider two cases for round $r$. 

\textbf{Case 1 [$\Phi^r \leq \frac{48}{\alpha^2} \cdot n$]:}  In this case, using \cref{eq:one_step_bound} we have
\[
\Pro{\Phi^{r+1} \leq e^{\alpha \log n} \cdot \frac{48}{\alpha^2} \cdot n  ~\left\vert \,~ \mathfrak{F}^r, \Phi^r \leq \frac{48}{\alpha^2} \cdot n \right.} \geq 1 - n^{-\omega(1)}.
\]

\textbf{Case 2 [$\frac{48}{\alpha^2} \cdot n < \Phi^r \leq e^{\alpha \log n} \cdot \frac{48}{\alpha^2} \cdot n$]:} Using \cref{lem:returns_to_stable} for $t_0 = r$,
\[
\Pro{\left. \bigcup_{s \in [r, r + c_s \cdot \frac{m^2}{n^2} \cdot \log n]} \left\{ \Phi^s \leq \frac{48}{\alpha^2} \cdot n \right\} \, ~\right|~ \, \mathfrak{F}^r, \frac{48}{\alpha^2} \cdot n < \Phi^r \leq e^{\alpha \log n} \cdot \frac{48}{\alpha^2} \cdot n } \geq 1 - n^{-7k}.
\]
Since we condition on $\frac{48}{\alpha^2} \cdot n < \Phi^r$, this is equivalent to
\begin{align} \label{eq:return_2}
\Pro{\left. \bigcup_{s \in [r+1, r + c_s \cdot \frac{m^2}{n^2} \cdot \log n]} \Phi^s \leq \frac{48}{\alpha^2} \cdot n \, ~\right|~ \, \mathfrak{F}^r, \frac{48}{\alpha^2} \cdot n < \Phi^r \leq e^{\alpha \log n} \cdot \frac{48}{\alpha^2} \cdot n } \geq 1 - n^{-7k}.
\end{align}

Let $t_0$ be an arbitrary but fixed round in $[t - c_r \cdot \frac{m^2}{n}, t]$ with $\Phi^{t_0} \leq \frac{48}{\alpha^2} \cdot n$ and $t_1 := t + m^2$.
Let $t_0 < r_1 <r_2<\cdots $ and $t_0 =: s_0<s_1<\cdots $ be two interlaced sequences (\cref{fig:green_red}) defined recursively for $i\geq 1$ by \[r_i := \min\left\lbrace r >  s_{i-1}: \frac{48}{\alpha^2} \cdot n< \Phi^r \leq e^{\alpha \log n} \cdot \frac{48}{\alpha^2} \cdot n \right\rbrace \quad\text{and} \quad s_i := \inf\left\lbrace s> r_i : \Phi^s \leq  \frac{48}{\alpha^2} \cdot n \right\rbrace. \] 
  Thus we have
  \[
  t_0 = s_0 < r_1 < s_1 < r_2 <s_2 < \cdots, 
  \]
  and since $r_i>r_{i-1}$ we have $ r_{t_1 - t_0}\geq t_1 - t_0$. 

\begin{figure}
    \centering
\scalebox{0.7}{
    \begin{tikzpicture}[
  IntersectionPoint/.style={circle, draw=black, very thick, fill=black!35!white, inner sep=0.05cm}
]

\definecolor{MyBlue}{HTML}{9DC3E6}
\definecolor{MyYellow}{HTML}{FFE699}
\definecolor{MyGreen}{HTML}{E2F0D9}
\definecolor{MyRed}{HTML}{FF9F9F}
\definecolor{MyDarkRed}{HTML}{C00000}

\node[anchor=south west] (plt) at (-0.1, 0) {\includegraphics{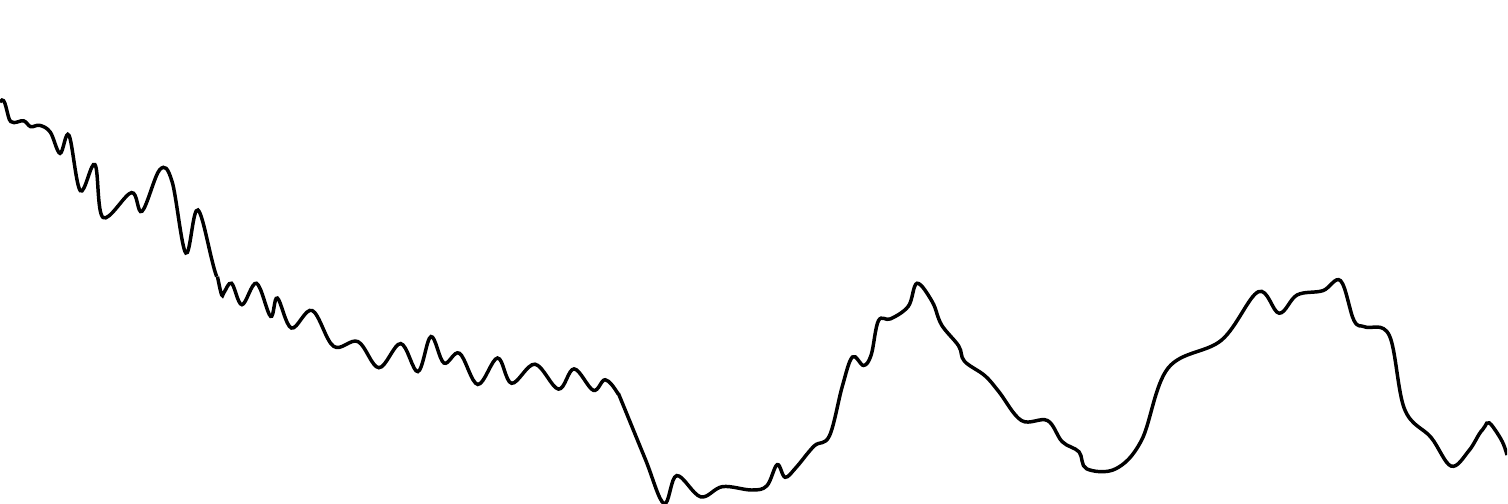}};

\def\xEnd{16}
\def\xLast{15.40}
\def\yLast{5}
\def\cn{1}
\def\yBottom{-0.8}

\node[anchor=south west, fill=MyBlue, rectangle,minimum width=6.42cm] at (0, \yLast) {Convergence phase};
\node[anchor=south west, fill=MyYellow, rectangle,minimum width=\xLast cm - 6.42cm] at (6.42, \yLast) {Stabilization phase};

\draw[dashed, thick] (0,1) -- (\xLast, 1);
\draw[dashed, thick] (0,1.8) -- (\xLast, 1.8);

\node[anchor=east] at (0, \cn) {$\frac{48}{\alpha^2} \cdot n$};
\node[anchor=east] at (0, 1.8) {$e^{\alpha \log n} \cdot \frac{48}{\alpha^2} \cdot n$};
\node[anchor=east] at (0, 4.35) {$e^{\alpha m}$};
\node[anchor=east] at (0, 5.3) {$\Phi^s$};

\node[anchor=west] at (\xEnd, \yBottom) {$s$};

\def\tA{6.42}
\def\tT{7.15}
\def\tB{8.52}
\def\tC{10.41}
\def\tD{11.71}
\def\tM{13.51}
\def\tE{14.36}

\newcommand{\drawLabel}[2]{
\draw[very thick] (#1, \yBottom) -- (#1, \yBottom -0.2);
\node[anchor=north] at (#1, \yBottom -0.15) {#2};}

\newcommand{\drawLine}[3]{
\draw[dashed, very thick, #3] (#1, \yBottom) -- (#1, \yLast);
\drawLabel{#1}{#2}}

\newcommand{\drawPoint}[3]{
\drawLine{#1}{#2}{#3}
\node[IntersectionPoint] at (#1, \cn) {};}

\draw[very thick] (0, \yBottom) -- (0, \yBottom -0.2);
\node[anchor=north] at (0, \yBottom -0.15) {$t - c_r \cdot \frac{m^2}{n}$};

\drawPoint{\tA}{$s_0$}{black!30!white};
\drawLabel{\tT}{$t$}{black!30!white};
\drawPoint{\tB}{$r_1$}{black!30!white};
\drawPoint{\tC}{$s_1$}{black!30!white};
\drawPoint{\tD}{$r_2$}{black!30!white};
\drawPoint{\tE}{$s_2$}{black!30!white};
\drawLine{\xLast}{$t + m^2$}{black!30!white};

\newcommand{\drawRegionRect}[3]{
\node[anchor=south west, rectangle, fill=#3, minimum width=#2 cm- #1 cm, minimum height=0.3cm] at (#1, \yBottom) {};}

\drawRegionRect{\tA}{\tB}{MyGreen}
\drawRegionRect{\tB}{\tC}{MyRed}
\drawRegionRect{\tC}{\tD}{MyGreen}
\drawRegionRect{\tD}{\tE}{MyRed}
\drawRegionRect{\tE}{\xLast}{MyGreen}

\newcommand{\drawBrace}[4]{
\draw [
    thick,
    decoration={
        brace,
        raise=0.5cm,
        amplitude=0.2cm
    },
    decorate
] (#2, \yBottom - 0.3) -- (#1, \yBottom - 0.3) 
node [anchor=north,yshift=-0.7cm,#4] {#3}; }

\drawBrace{0}{\tA}{Convergence by \cref{lem:convergence}}{pos=0.5};
\drawBrace{\tB}{\tC}{}{};
\drawBrace{\tD}{\tE}{Each $s_i - t_i \leq c_s \cdot \frac{m^2}{n^2} \cdot \log n$ by \cref{thm:stabilization} }{};

\draw[->, ultra thick] (0,\yBottom) -- (0, 6);
\draw[->, ultra thick] (0,\yBottom) -- (\xEnd, \yBottom);

\end{tikzpicture}}

    \caption{Visualization of the $s_i$/$r_i$ interlacing. The \Whp~existence of $s_0$ follows from \cref{eq:exists_t0}.}
    \label{fig:green_red}
\end{figure}
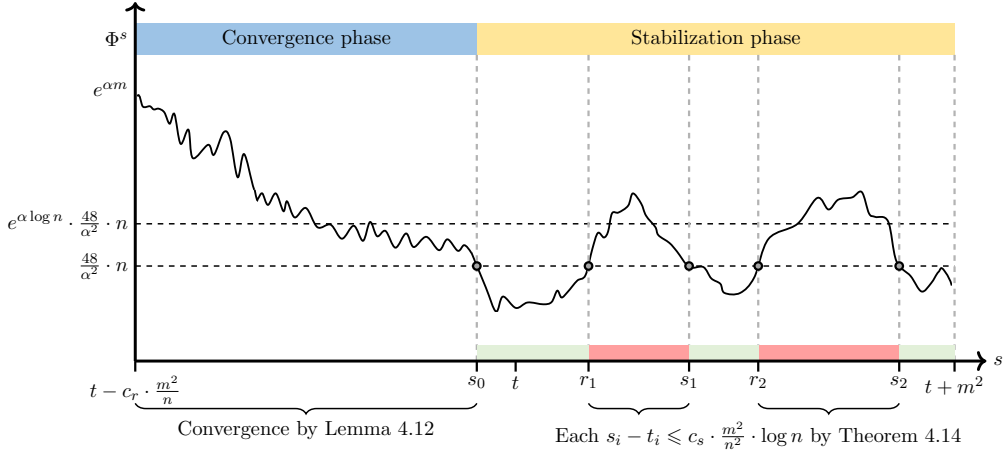

By \cref{eq:return_2} we have for any $i=1,2,\ldots, t_1 - t_0$ and any $r = t_0 + 1, \ldots , t_1$ 
\[
  \Pro{\left. s_i-r_i> c_s \cdot \frac{m^2}{n^2} \cdot \log n ~\right|~ \mathfrak{F}^{r},  \frac{48}{\alpha^2} \cdot n < \Phi^r \leq e^{\alpha \log n} \cdot \frac{48}{\alpha^2} \cdot n,  r_i = r } \leq n^{-7k}.
\]

Since the above bound holds for any $i$ and $\mathfrak{F}^{r}$, with $r_i=r$, it follows by the union bound over all $i=1,2,\ldots,t_1 - t_0$, as $t_1 - t_0 \leq 2m^2$,
\begin{align}
\Pro{\left. \bigcap_{i=1}^{t_1 - t_0} \left\lbrace s_i-r_i\leq c_s \cdot \frac{m^2}{n^2} \cdot \log n \right\rbrace \, \right\vert \, \mathfrak{F}^{t_0}, \Phi^{t_0} \leq \frac{48}{\alpha^2} \cdot n } & \geq 1 - n^{-7k} \cdot (2m^2) \\
& \geq 1- 2n^{-5k}. \label{eq:close_step}
\end{align}
Furthermore, taking the union bound over \cref{eq:one_step_bound},
\begin{align*}
\lefteqn{ \Pro{ \bigcap_{r=t_0}^{t_1} \left\{ \Phi^{r+1} \leq e^{\alpha \log n} \cdot \frac{48}{\alpha^2} \cdot n \cap  \Phi^{r} \leq \frac{48}{\alpha^2} \cdot n \right\} \, \left\vert \, \mathfrak{F}^{t_0}, \Phi^{t_0} \leq \frac{48}{\alpha^2} \cdot n \right.}} \\
 &\geq 1 - 
 \sum_{r=t_0}^{t_1} \Pro{ \left. \left\{ \Phi^{r+1} > e^{\alpha \log n} \cdot \frac{48}{\alpha^2} \cdot n \right\} \cap \left\{  \Phi^{r} \leq \frac{48}{\alpha^2} \cdot n \right\} \, \right\vert \, \mathfrak{F}^{t_0}, \Phi^{t_0} \leq \frac{48}{\alpha^2} \cdot n } \\
 &\geq 1 - \sum_{r=t_0}^{t_1} \Pro{ \Phi^{r+1} > e^{\alpha \log n} \cdot \frac{48}{\alpha^2} \cdot n ~\Bigg|~ \Phi^{r} \leq \frac{48}{\alpha^2} \cdot n } \cdot \Pro{\left. \Phi^{r} \leq \frac{48}{\alpha^2} \cdot n \, \right\vert \, \mathfrak{F}^{t_0}, \Phi^{t_0} \leq \frac{48}{\alpha^2} \cdot n }
 \\ &\geq 1 - (t_1-t_0) \cdot n^{-\omega(1)} \geq 1 - n^{-\omega(1)}.
\end{align*}
Conditioning on the above event occurring, we know that during any interval $[s_i,\min\{r_{i+1},t_1\}]$ the potential is small, i.e., for any $i \geq 1$,
\begin{align}
 \Pro{ \left. \bigcap_{r \in [s_i,\min\{r_{i+1},t_1\}]} \left\{ \Phi^{r} \leq \frac{48}{\alpha^2} \cdot n \right\} \,\,\right\vert\,\, \mathfrak{F}^{t_0}, \Phi^{t_0} \leq \frac{48}{\alpha^2} \cdot n } \geq 1 -n^{-\omega(1)}. \label{eq:take_two}
\end{align}
Combining \cref{eq:close_step} and \cref{eq:take_two}, it follows that $\Phi^{r} \leq e^{\alpha \log n} \cdot \frac{48}{\alpha^2} \cdot n$ holds at least every $c_s \cdot \frac{m^2}{n^2} \cdot \log n$ rounds, i.e.,
\begin{align}
 \Pro{ \left. \bigcap_{r \in [t_0,t_1-c_s \cdot \frac{m^2}{n^2} \cdot \log n ]} 
 \bigcup_{\rho \in [0,c_s \cdot \frac{m^2}{n^2} \cdot \log n ]}
 \left\{ \Phi^{r+\rho} \leq e^{\alpha \log n} \cdot \frac{48}{\alpha^2} \cdot n \right\} \,\,\right\vert\,\, \mathfrak{F}^{t_0}, \Phi^{t_0} \leq \frac{48}{\alpha^2} \cdot n} \geq 1 -3n^{-5k}. \label{eq:take_three}
\end{align}

Define $\tau := \min \{ u \geq t_1 - c_s \cdot \frac{m^2}{n^2} \cdot \log n \colon \Phi^{u} \leq e^{\alpha \log n} \cdot \frac{48}{\alpha^2} \cdot n \}$. Then by \cref{eq:take_three},
\begin{align}
\Pro{\tau \leq t_1 \, \left\vert \, \mathfrak{F}^{t_0}, \Phi^{t_0} \leq \frac{48}{\alpha^2} \cdot n \right. } \geq 1 - 3n^{-5k}. \notag%
\end{align}
Further, let $\tilde{t}_0 := \inf\{ \rho \geq t - c_r \cdot \frac{m^2}{n} : \Phi^{\rho} \leq \frac{48}{\alpha^2} \cdot n \}$. Then,
\begin{align}
\Pro{\tau \leq t_1} & = \sum_{t_0 = t - c_r \cdot \frac{m^2}{n}}^t \Pro{\tau \leq t_1 \, \left\vert \, \mathfrak{F}^{t_0}, \Phi^{t_0} \leq \frac{48}{\alpha^2} \cdot n, \tilde{t}_0 = t_0 \right. } \cdot \Pro{\tilde{t}_0 = t_0} \notag \\
 & \geq (1 - 3n^{-5k}) \cdot \sum_{t_0 = t - c_r \cdot \frac{m^2}{n}}^t \Pro{\tilde{t}_0 = t_0} \notag \\
 &\stackrel{(a)}{\geq} (1 - 3n^{-5k}) \cdot (1 - e^{-\Omega(n)}) \notag \\
 & \geq 1 - 4n^{-5k}, \label{eq:take_four}
\end{align}
where in $(a)$ we used \cref{eq:exists_t0}.

In the round $\tau$, which by definition satisfies $\Phi^{\tau} \leq e^{\alpha \log n} \cdot \frac{48}{\alpha^2}$, it follows that the maximum load is at most
\[
\max_{i \in [n]}x_i^{\tau} \leq \frac{1}{\alpha} \cdot \log\Big( e^{\alpha \log n} \cdot \frac{48}{\alpha^2} \cdot n \Big) \leq \tilde{c}_1 \cdot \frac{m}{n} \cdot \log n,
\]
for some constant $\tilde{c}_1 > 0$, since $m \leq n^k$. In the interval $[\tau,t_1]$, we will re-allocate at most $\tilde{m}:= c_s \cdot \frac{m^2}{n} \cdot \log n$ balls. Each of these balls will be re-allocated to a bin chosen uniformly and independently at random from $[n]$. By a Chernoff bound, with probability at least $1-n^{-10k}$, the maximum load when allocating $\tilde{m}$ balls into $n$ bins is at most
\begin{align}
 \frac{\tilde{m}}{n} + 20 \cdot k \cdot \sqrt{\frac{\tilde{m}}{n} \cdot \log n }
 \leq c_s \cdot \frac{m^2}{n^2} \cdot \log n + 20 \cdot k \cdot \sqrt{ c_s \cdot \frac{ m^2}{n^2} \cdot \log^2 n } \leq \tilde{c_2} \cdot \frac{m}{n} \cdot \log n. \label{eq:take_five}
\end{align}
Since in this case, $
 \max_{i \in [n]} x_i^{t_1} \leq \max_{i \in [n]} x_i^{\tau} + \tilde{c_2} \cdot \frac{m}{n} \cdot \log n$, the union bound over \cref{eq:take_four} and \cref{eq:take_five} yields
\begin{align*}
\Pro{\max_{i \in [n]} x_i^{t_1} \leq (\tilde{c}_1 + \tilde{c}_2) \cdot \frac{m}{n} \cdot \log n} \geq 1 - 4n^{-5k} - n^{-10k} \geq 1 - 5n^{-5k}.
\end{align*}
Finally, taking the union bound over all possible rounds $t_1 \in [t,t+m^2]$ yields the statement of the theorem.
\end{proof}

\section{The Multi-Token Traversal Time} \label{sec:traversal}

As mentioned in~\cite{BCNPP19}, it is natural to regard the RBB process as a multi-token traversal problem, in which each ball should visit all bins as frequently as possible. This can be seen as a ``cover time'' of parallel and dependent random walks, which is the first time until each ball has been allocated at least once to every bin. In~\cite[Corollary~1]{BCNPP19}, a \Whp~bound of $\Oh(n \log^2 n)$ on this quantity was established  (it was also shown that this bound holds even in an adversarial setting, where an adversary is able to re-allocate all tokens arbitrarily every $\Oh(n)$ rounds). For the original setting without the adversary, we give the following improvement:

\begin{pro}\label{pro:traversal}
Consider the RBB with any $m \geq n$. Then, with probability $1-m^{-2}$, each of the $m$ balls traverses all $n$ bins within $28 m \cdot \log m$ rounds. Furthermore, any fixed ball needs with probability at least $1-o(1)$ at least $1/16 \cdot m \cdot \log n$ rounds until all $n$ bins are traversed.
\end{pro}
Note that for $m=\poly(n)$, the upper and lower bounds match each other up to constant factors.
\begin{proof}
For the upper bound, we assume that balls are inserted into bins using the FIFO scheme, with ties among the balls being inserted into the same bin in the same round broken arbitrarily. 

Fix a ball, w.l.o.g., ball $1$, and let us denote by $1 \leq t_1 < t_2 < \ldots < t_{\tau(1)}$ the sequence of time-steps in which the ball is re-allocated to another bin, and let $\tau(1) \geq n-1$ be the number of re-allocations (switches) until ball $1$ has been in all bins. Recall that the random variable $\tau(1)$ is also known as the coupon collector problem~\cite{MU17}, and we can use the following standard bound:
\begin{align}
  \Pro{ \tau(1) \geq 4 n \cdot \log m } \leq 
  n \cdot \left(1 - \frac{1}{n} \right)^{4 n \cdot \log m}
 \leq  m^{-3}, \label{eq:coupon}
\end{align}
where in the last inequality we used $m \geq n$.
Every time-step $t_i$, ball $1$ switches to its $i$-th bin, it will be in the same bin as some other ball $j \neq 1$ with probability $\frac{1}{n}$. This other ball $j$ (if it is not reallocated), will be in front of ball $i$ with probability equal $\frac{1}{n}$; if $j$ is also reallocated in round $t_i$, then this probability is at most $\frac{1}{n}$. Thus the total delay that ball $j$ causes to bin $1$ before $4 n \cdot \log m+1$ switches can be bounded from above by
\[
  X(j) := \sum_{k=1}^{4n \cdot \log m} X_k,
\]
where the $X_k \in \{0,1\}$ are independent Bernoulli random variables with parameter $1/n$. Hence $\Ex{X(j)} = 4 \log m$. Using a Chernoff Bound for binomial random variables,~\cite[Theorem~4.4]{MU17}, for any $R \geq 6 \cdot \Ex{X(j)}$,
\[
 \Pro{ X(j) \geq R } \leq 2^{-R},
\]
and thus with $R=24 \log m$,
\[
 \Pro{ X(j) \geq 24 \log m} \leq m^{-10}.
\]
Taking the union bound over all other $m-1$ balls $j \in [m] \setminus \{1\}$,
\[
 \Pro{ \bigcup_{j \in [m] \setminus \{1\}} \left\{ X(j) \geq 24 \log m \right\} }\leq m^{-9}.
\]
Let $X:=\sum_{j \in [m] \setminus \{1\}} X(j)$, which is the total delay caused by other balls until ball $1$ achieves  $4 n \cdot \log m$ switches. By the above,
\begin{align}
\Pro{ X \geq 24 m \cdot \log m  }\leq m^{-9}. \label{eq:total_wait}
\end{align}
Combining \cref{eq:coupon} and \cref{eq:total_wait}, we conclude by the union bound for ball $1$,
\begin{align*}
\lefteqn{  \Pro{ t_{\tau(1)} \leq 24 m \cdot \log m + 4 n \cdot \log m  } } \\
 &\geq \Pro{ \left\{ \tau(1) \geq 4n \cdot \log m \right\} \cap \left\{ X \leq 24 m \cdot \log m \right\} } \\
 &\geq 1 - m^{-3} - m^{-9} \geq 1 - m^{-2}.
\end{align*}
The upper bound now follows by another union bound over all balls $i \in [n]$.

We now turn to the lower bound. Again, we assume that balls are allocated using the FIFO scheme, but now we also assume that ties among the balls reaching the same bin in the same round are broken randomly.

Again, first consider ball $1$ w.l.o.g. We may assume that $t_1=1$, meaning that ball $1$ switches in the very first round. For the coupon collector problem, one can derive the following lower bound based on Chebyshev's inequality~\cite[Chapter~3]{MU17},
\begin{align}
  \Pro{ \tau(1) \leq 1/2 \cdot n \cdot \log n } \leq \frac{1}{\log n}. \label{eq:coupon_two}
\end{align}
Following the arguments in the upper bound, the total delay any fixed ball $j \neq 1$ causes to bin $1$ before $1/2 \cdot n \cdot \log n$ switches, can be bounded from below by
\[
 Y(j) := \sum_{k=1}^{1/2 \cdot n \cdot \log n} Y_k,
\]
where the $Y_k \in \{0,1\}$ are independent Bernoulli random variables with parameter $1/(2n)$ (in the pessimistic case for the lower bound, ball $j$ and $1$ are always re-allocated in the same round). Hence $\Ex{Y(j)} = 1/4 \cdot \log n$. 
Applying the Chernoff bound~\cite[Theorem 4.5]{MU17},
\[
 \Pro{ Y(j) \leq (1-\delta) \cdot \Ex{Y(j)} } \leq \exp\left( -\delta^2/ 2 \cdot \Ex{Y(j)} \right),
\]
with $\delta=1/2$, we obtain
\[
 \Pro{ Y(j) \leq 1/8 \cdot \log n} \leq n^{-1/32}.
\]
Next define $\mathcal{I}:=\left\{ j \in [m] \setminus \{1\} \colon Y(j) \leq 1/8 \cdot \log n \right\}$. Then $\Ex{ | \mathcal{I } | } \leq m \cdot n^{-1/32}$, and by Markov's inequality,
\[
 \Pro{ | \mathcal{I} | \geq m/2 } \leq n^{-1/33}.
\]

Let $Y:=\sum_{j \in [m] \setminus \{1\}} Y(j)$ be the delay caused by all other $m-1$ balls before ball $1$ makes $1/2 \cdot n \cdot \log n+1$ switches. Then by the above,
\begin{align}
 \Pro{ Y \leq m/2 \cdot 1/8 \cdot \log n } \leq n^{-1/33} \label{eq:total_wait_two}.
\end{align}
By the union bound of \cref{eq:coupon} and \cref{eq:total_wait_two},
\[
 \Pro{ t_{\tau(1)} \geq 1/16 \cdot m \cdot \log n} \geq 1 - \frac{1}{\log n} - n^{-1/33} \geq 1-o(1). \qedhere
\]
\end{proof}

\section{Experiments} \label{sec:experiments}

We complement our analysis with some experimental results (\cref{fig:avg_gap}, \cref{fig:fraction_non_empty} and \cref{fig:convergence_time}).

In \cref{fig:avg_gap}, we plot the maximum load vs the average number of balls for $n \in \{ 10^2, 10^3, 10^4 \}$ and $m \in \{ n, 2n, \ldots 50n \}$ after $10^6$ rounds starting with the uniform distribution. The trend seems to be linear in $m/n$ as $m$ grows, which is in line with the $\Theta(m/n \cdot \log n)$ bound on the maximum load shown by our theoretical analysis in \cref{lem:lower_bound} and \cref{thm:stabilization}. In \cref{fig:fraction_non_empty}, we plot the fraction of empty bins vs the average number of balls for $n \in \{ 10^2, 10^3, 10^4 \}$ and $m \in \{ n, 2n, \ldots, 50n \}$ averaged over $10^6$ rounds, starting from the uniform load vector. The trend supports that the fraction is $\Theta(n/m)$ in steady state, as proven in \cref{lem:quadratic_eps,lem:many_empty}.

\begin{figure}[H]
    \centering
    \includegraphics[scale=0.9]{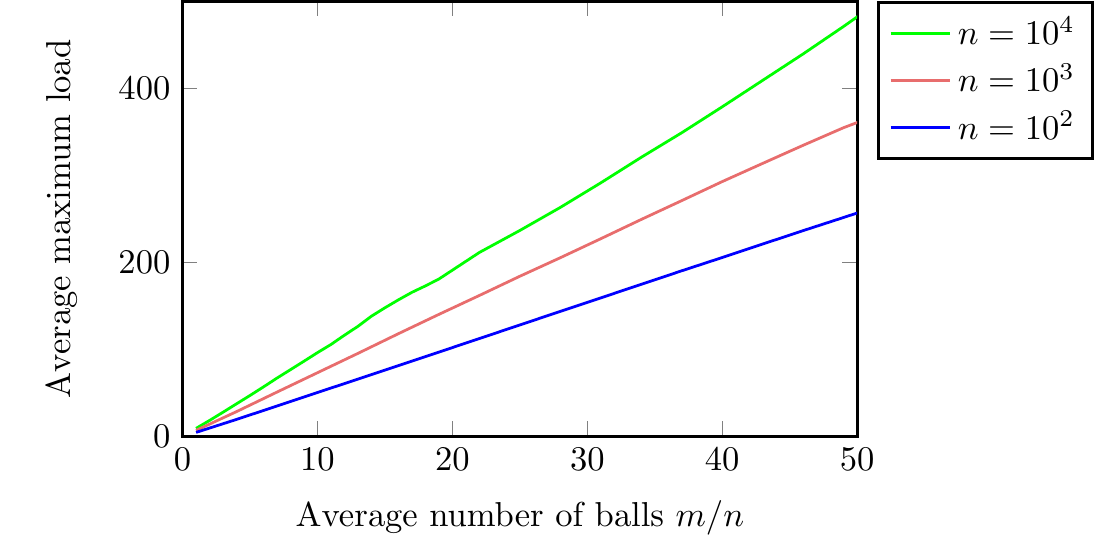}
    \caption{Maximum load for $n \in \{ 10^2, 10^3, 10^4 \}$ and $m \in \{ n, 2n, \ldots 50n \}$ after $10^6$ rounds, starting from the uniform load vector (averaged over $25$ runs). These plots support the $\Theta(m/n \cdot \log n)$ maximum load, as suggested by our theoretical analysis (\cref{lem:lower_bound} and \cref{thm:stabilization}).}
    \label{fig:avg_gap}
\end{figure}

\begin{figure}[H]
    \centering
    \includegraphics[scale=0.9]{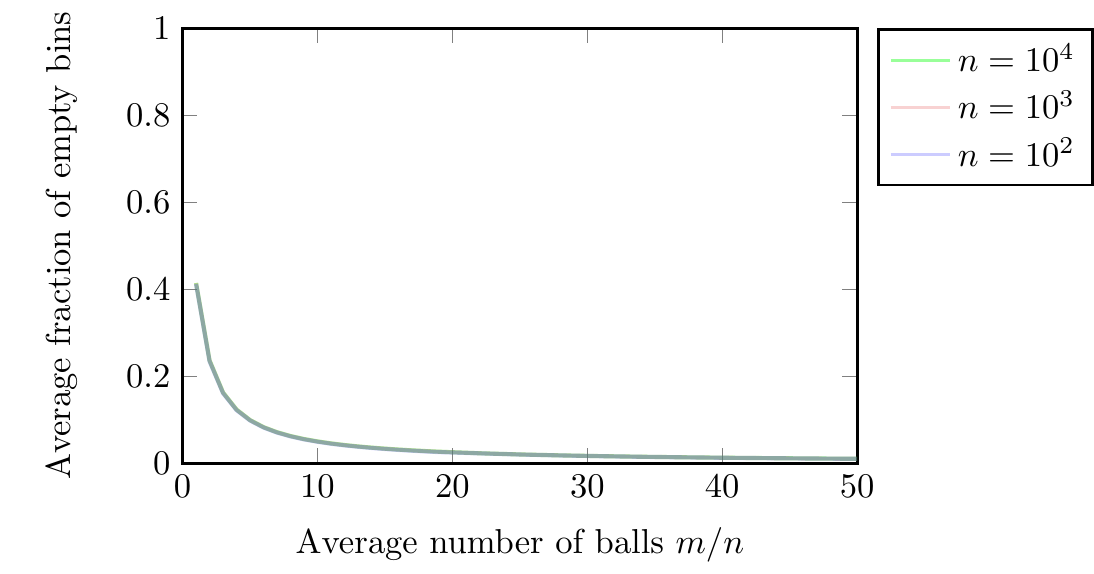}
    \caption{Fraction of empty bins vs the average load for $n \in \{ 10^2, 10^3, 10^4 \}$ and $m \in \{ n, 2n, \ldots , 50n \}$ averaged over $10^6$ rounds, starting from the uniform load vector (averaged over $25$ runs). Note that for all values of $n$, the curves are very close to one another.}
    \label{fig:fraction_non_empty}
\end{figure}

\begin{figure}[H]
    \centering
    \includegraphics[scale=0.9]{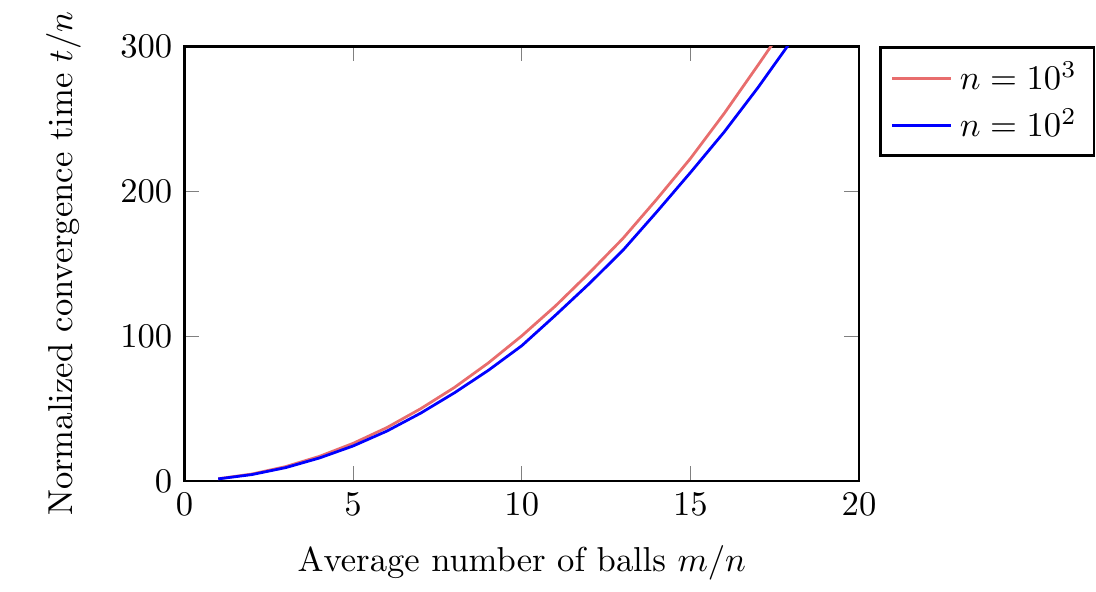}
    \caption{Convergence time for $n \in \{ 10^2, 10^3 \}$ to reach a configuration with maximum load $\leq 1.5 \cdot \frac{m}{n} \cdot \log n$, starting from the $(m, 0, \ldots, 0)$ load vector (averaged over $100$ runs). These suggest that the convergence time is $\Omega(m^2/n)$, which complements our theoretical bound of $\Oh(m^2/n)$ (\cref{lem:convergence}).}
    \label{fig:convergence_time}
\end{figure}

\section{Conclusions}\label{sec:conclusions}

We revisited the RBB process and proved that for any $m \geq n$ \Whp~after $\Oh(m^2/n)$ rounds it achieves an $\Oh(m/n \cdot \log m)$ maximum load. For $n \leq m \leq \poly(n)$ we show that it stabilizes in a configuration with an $\Oh(m/n \cdot \log n)$ maximum load, for at least $m^2$ rounds and also prove a lower bound matching up to multiplicative constants. This resolved two conjectures in~\cite{BCNPP19}. We also obtained an upper bound of $\Oh(m \cdot \log m)$ on the traversal time for the balls, which was shown to be tight for any $m=\poly(n)$.

There are several possible extensions, such as generalizing the stabilization result for $m = n^{\omega(1)}$, determining whether the $\Oh(m^2/n)$ convergence time is tight for $m = \omega(n)$ and determining tight bounds for the maximum load when $m < n$. 

Finally, as mentioned in~\cite{BCNPP19}, an interesting but also challenging generalization is the RBB process on graphs. We hope that at least some of our arguments could be leveraged, for example, the insight in \cref{lem:many_empty} that many bins become empty within $\Oh((m/n)^2)$ rounds might extend to graphs.

\bibliographystyle{ACM-Reference-Format-CAM}
\bibliography{bibliography}

\clearpage

\appendix

\section{Tools}

\subsection{Facts about the \OneChoice Process}

In this section, we prove several basic facts about the \OneChoice process. We start by upper bounding the quadratic potential. 

\begin{lem} \label{lem:one_choice_upsilon}
For the \OneChoice process for $n$ balls into $n$ bins, for any $t \geq 0$ \[
\Pro{\Upsilon^t \leq 3 \cdot n} \geq 1 - n^{-\omega(1)}.
\]
\end{lem}
\begin{proof}
Recall
\[
 \Upsilon^t := \sum_{i=1}^n \Upsilon_i^t = \sum_{i=1}^n (x_i^t)^2.
\]
Since $x_i^t$ has distribution $\Bin(t,1/n)$ for any $t \geq 0$, $\Ex{ (x_i^n)^2 } 
= (1-1/n) + 1 = 2 - \frac{1}{n}$.
Hence by linearity of expectations,
\[
\Ex{ \Upsilon^n } \leq 2n.
\]
Define 
\[
 \tilde{\Upsilon}^n := \sum_{i=1}^n \min\{ \Upsilon_i^n, \log^2 n\},
\]
and note that $\Upsilon^n = \tilde{\Upsilon}^n$ if and only if the maximum load is at most $\log n$.
Then, 
\[
\Ex{\tilde{ \Upsilon }^n }  \leq \Ex{ \Upsilon^n } \leq 2n.
\]
Further, $\tilde{\Upsilon}^n$ is a function of $n$ independent random variables (the random bin choices of the $n$ balls), and changing one of these choices can change $\tilde{\Upsilon}^n$ by at most $(\log n)^2 - (\log n - 1)^2 \leq 2 \log n$. Hence by the Method of Bounded Differences (\cref{thm:mobd}),
\begin{align*}
 \Pro{ \tilde{ \Upsilon }^n  - \Ex{\tilde{ \Upsilon }^n } \geq \lambda } \leq \exp\left( - \frac{ \lambda^2}{2 \sum_{i=1}^n 4 (\log n)^2 } \right),
\end{align*}
and choosing $\lambda = n$ yields,
\begin{align*}
\Pro{ \tilde{ \Upsilon }^n \geq 3n} \leq \Pro{ \tilde{ \Upsilon }^n \geq \Ex{\tilde{ \Upsilon }^n }  + n} \leq n^{-\omega(1)}.
\end{align*}
Further, since the maximum load is larger than $\log n$ with probability $1-n^{-\omega(1)}$, we have by the union bound
\[
\Pro{ \Upsilon^n \geq 3n} 
 \leq \Pro{ \left\{\tilde{ \Upsilon }^n \geq 3n \right\} \cup \left\{ \max_{i \in [n]} x_i^{n} > \log n  \right\} }  
\leq n^{-\omega(1)} + n^{-\omega(1)} = 2 n ^{-\omega(1)}.
\qedhere
\]
\end{proof}

The next standard result was also used in~\cite[Section~4]{PTW15} and is based on~\cite{RS98}. For convenience of the reader, we give a self-contained proof, obtaining high probability bounds.

\begin{lem} \label{lem:one_choice_cnlogn}
\OneChoiceCnlogn
\end{lem}
\begin{proof}
In order to use the Poisson Approximation~\cite[Chapter 5]{MU17}, 
let $Y_1,Y_2,\ldots,Y_n$ be $n$ independent Poisson random variables with parameter $\lambda=\frac{m}{n}=c \log n$. Then,
\begin{align*}
 \Pro{ Y_i \geq \lambda + \frac{\sqrt{c}}{10} \cdot \log n} &\geq  \Pro{ Y_i = \lambda + \frac{\sqrt{c}}{10} \cdot \log n} \\ 
 &= e^{-\lambda} \cdot \frac{ \lambda^{\lambda + \frac{\sqrt{c}}{10} \cdot \log n}}{ (\lambda + \frac{\sqrt{c}}{10} \cdot \log n)!}.
\end{align*}
Using that $z! \leq \sqrt{2 \pi z} \left( \frac{z}{e} \right)^{z} e^{\frac{1}{12z}}$ for any integer $z \geq 1$,
\begin{align*}
 \Pro{ Y_i = \lambda + \frac{\sqrt{c}}{10} \cdot \log n} &\geq \frac{1}{4 \cdot \sqrt{2 \pi \lambda}} \cdot e^{-\lambda} \cdot 
\left( \frac{e \lambda }{ \lambda + \frac{\sqrt{c}}{10} \cdot \log n} \right)^{ \lambda + \frac{\sqrt{c}}{10} \cdot \log n} \\
&\geq \frac{1}{4 \cdot \sqrt{2 \pi \lambda}} \cdot e^{\frac{\sqrt{c}}{10} \log n} \cdot \left( 1 + \frac{1}{10 \sqrt{c}} \right)^{-\lambda - \frac{\sqrt{c}}{10} \cdot \log n} \\
&\geq \frac{1}{4 \cdot \sqrt{2 \pi \lambda}} \cdot e^{\frac{\sqrt{c}}{10} \log n} \cdot e^{-\frac{1}{10 \sqrt{c}} \cdot (\lambda +\frac{\sqrt{c}}{10} \cdot \log n)} \\
&\geq \frac{1}{4 \cdot \sqrt{2 \pi \lambda}} \cdot e^{\frac{\sqrt{c}}{10} \log n-\frac{1}{10 \sqrt{c}} \lambda-\frac{1}{100} \log n} \\
&\geq \frac{1}{4 \cdot \sqrt{2 \pi \lambda}} \cdot e^{-\frac{1}{100} \log n} 
 \end{align*}
 Since for any $k \geq 0$,
 \begin{align*}
 \frac{\Pro{ Y_i = k+1}}{ \Pro{ Y_i = k} } &= \frac{\lambda}{k+1},
 \end{align*}
 we conclude that
 \begin{align*}
 \Pro{ Y_i \geq \lambda + \frac{\sqrt{c}}{10} \cdot \log n} &\geq \sum_{k=0}^{\sqrt{\lambda}-1} \Pro{ Y_i = \lambda + \frac{\sqrt{c}}{10} \cdot \log n + k} \\
 &\geq \sqrt{\lambda} \cdot \Pro{ Y_i = \lambda + \frac{\sqrt{c}}{10} \cdot \log n + \sqrt{\lambda}} \\
 &\geq  \sqrt{\lambda} \cdot \Pro{ Y_i = \lambda + \frac{\sqrt{c}}{10} \cdot \log n} \cdot \prod_{k=1}^{\sqrt{\lambda}} \left( \frac{\lambda}{\lambda+\frac{\sqrt{c}}{10} \cdot \log n+k} \right) \\
  &\geq  \sqrt{\lambda} \cdot \frac{1}{4 \cdot \sqrt{2 \pi \lambda}} \cdot e^{-\frac{1}{100} \log n} \cdot \left( \frac{\lambda}{\lambda+\frac{\sqrt{c}}{10} \cdot \log n+\sqrt{\lambda}} \right)^{\sqrt{\lambda}} \\
  &\geq  \sqrt{\lambda} \cdot \frac{1}{4 \cdot \sqrt{2 \pi \lambda}} \cdot e^{-\frac{1}{100} \log n} \cdot \left( 1 + \frac{1}{5 \sqrt{c}}  \right)^{-\sqrt{\lambda}} \\
 &\geq  \sqrt{\lambda} \cdot \frac{1}{4 \cdot \sqrt{2 \pi \lambda}} \cdot e^{-\frac{1}{100} \log n} \cdot e^{-\frac{1}{5} \sqrt{\log n}} \\
 &\geq e^{-\frac{1}{99} \log n} = n^{-1/99},
 \end{align*}
 where the last inequality holds for sufficiently large $n$. Hence,
 \begin{align*}
  \Pro{ \bigcup_{i=1}^n \left\{ Y_i \geq \lambda + \frac{\sqrt{c}}{10} \cdot \log n \right\} } & \geq
  1 - \left(  1 - n^{-1/99} \right)^{n} \geq 1 - n^{-3}.
 \end{align*}
Hence for $\tilde{\mathcal{E}}:= \left\{\max_{i \in [n]} Y_i \geq \lambda + \frac{\sqrt{c}}{10} \cdot \log n \right\}$, we have $\Pro{\neg \tilde{\mathcal{E}} } \leq n^{-3}$. Note that  $\tilde{\mathcal{E}}$ is a monotone event under adding balls, and thus with $\mathcal{E}:= \left\{\max_{i \in [n]} y_i^{m} \geq \lambda + \frac{\sqrt{c}}{10} \cdot \log n \right\}$, we have by~\cite[Corollary 5.11]{MU17})
\[
 \Pro{\neg \mathcal{E} } \leq 2 \cdot \Pro{ \neg \tilde{\mathcal{E}} } \leq 2  \cdot n^{-3} \leq n^{-2}. \qedhere
\]
\end{proof}

\subsection{Auxiliary Probabilistic Inequalities}

In this section, we prove two simple probabilistic inequalities.

\begin{lem}\label{lem:binomial_bound}
For any $n \geq 8$ and $\gamma \in [n]$, 
\[
\Pro{\Bin(n, 1/n) = \gamma} \leq 2^{-\gamma}.
\]
\end{lem}
\begin{proof}
Note that for $\gamma \geq 4$, we have $\gamma! \geq 2^\gamma$. For $n \geq 8$, we also have that $e^{1 - \frac{\gamma}{n}} \cdot \gamma! \geq 2^{\gamma}$ for $\gamma < 4$, and so for all $\gamma \in [n]$. 

Hence,
\begin{align*}
\Pro{\Bin(n, 1/n) = \gamma} & = \binom{n}{\gamma} \cdot \frac{1}{n^{\gamma}} \cdot \Big(1 - \frac{1}{n}\Big)^{n - \gamma} \\
 & = \frac{n}{n} \cdot \frac{n-1}{n} \cdot \ldots \cdot \frac{n- \gamma + 1}{n} \cdot \frac{1}{\gamma!} \cdot \Big( 1 - \frac{1}{n} \Big)^{n - \gamma} \\
 & \leq \frac{1}{\gamma!} \cdot \Big( 1 - \frac{1}{n} \Big)^{n - \gamma} \\
 & \leq \frac{1}{\gamma!} \cdot e^{-1 + \gamma/n} \leq 2^{-\gamma},
\end{align*}
using that $1+z \leq e^z$ for any $z$ and $\gamma \leq n$.
\end{proof}

\begin{lem} \label{lem:geometric_arithmetic}
Consider a sequence of random variables $(X^i)_{i \in \mathbb{N}}$ such that there are $0 < a < 1$ and $b > 0$ such that every $i \geq 1$,
\[
\Ex{X^i \mid X^{i-1}} \leq X^{i-1} \cdot a + b.
\]
Then for every $i \geq 1$, 
\[
\Ex{X^i \mid X^0}
\leq X^0 \cdot a^i + \frac{b}{1 - a}.
\]
\end{lem}
\begin{proof}
We will prove by induction that for every $i \in \mathbb{N}$, 
\[
\Ex{X^i \mid X^0} \leq X^0 \cdot a^i + b \cdot \sum_{j = 0}^{i-1} a^j.
\]
For $i = 0$, $\Ex{X^0 \mid X^0} \leq X^0$. Assuming the induction hypothesis holds for some $i \geq 0$, then since $a > 0$,
\begin{align*}
\Ex{X^{i+1} \mid X^0} & = \Ex{\Ex{X^{i+1} \mid X^i}\mid X^0} \leq \Ex{X^i\mid X^0} \cdot a + b \\
 & \leq \Big(X^0 \cdot a^i + b \cdot \sum_{j = 0}^{i-1} a^j \Big) \cdot a + b \\
 & = X^0 \cdot a^{i+1} +b \cdot \sum_{j = 0}^i a^j.
\end{align*}
The claims follows using that for $a \in (0,1)$, $\sum_{j=0}^{\infty} a^j = \frac{1}{1-a}$.
\end{proof}

\subsection{Simple Bounds for the RBB process}

We will now use the upper bound for the quadratic potential for \OneChoice (\cref{lem:one_choice_upsilon}) to upper bound the maximum change of the quadratic potential for the RBB process over one round.

\begin{lem} \label{lem:bound_upsilon_change}
Consider the RBB process with $m \geq n$ balls and $n$ bins. For any round $t \geq 0$, such that $\max_{i \in [n]} x_i^t \leq \frac{m}{n} \cdot \log n$, we have,
\[
\Pro{|\Upsilon^{t+1} - \Upsilon^t| \leq 2 \cdot m \cdot \log n + 3n} \geq 1 - n^{-\omega(1)}.
\]
\end{lem}
\begin{proof}
Let $k_i \in [0, n]$ be the number of balls that each bin receives at round $t$. For any bin $i \in [n]$ with $x_i^t > 0$, 
\begin{align*}
|\Upsilon_i^{t+1} - \Upsilon_i^t| 
 & = |(x_i^t + k_i - 1)^2 - (x_i^t)^2| = |2 \cdot (k_i - 1) \cdot x_i^t + (k_i- 1)^2| \\
 & \leq 2 \cdot x_i^t \cdot k_i+ (k_i)^2 + 1.
\end{align*}
For any bin $i \in [n]$ with $x_i^t = 0$,
\[
|\Upsilon_i^{t+1} - \Upsilon_i^t| = k_i^2.
\]
Aggregating over all bins, we have
\begin{align} \label{lem:upsilon_upper_bound}
|\Upsilon^{t+1} - \Upsilon^t| & 
 \leq  \sum_{i = 1}^n 2 \cdot x_i^t \cdot k_i+ \sum_{i=1}^n (k_i)^2 + n.
\end{align}
Using \cref{lem:one_choice_upsilon} we have
\[
\Pro{\sum_{i = 1}^n k_i^2 \leq 3 \cdot n} \geq 1 - n^{-\omega(1)}.
\]
Assuming that the event $\sum_{i = 1}^n k_i^2 \leq 3 \cdot n$ holds, and since by assumption $\max_{i \in [n]} x_i^t \leq \frac{m}{n} \cdot \log n$ and $m \geq n$, we finally conclude from \cref{lem:upsilon_upper_bound}
\[
|\Upsilon^{t+1} - \Upsilon^t| \leq 2 \cdot n \cdot \frac{m}{n} \cdot \log n + 3n = 2 \cdot m \cdot \log n + 3n. \qedhere
\]
\end{proof}

\subsection{Martingale and Drift Inequalities}

\begin{thm}[Corollary 5.2 in~\cite{DP09}] \label{thm:mobd}
Consider a function $f : \prod_{i \in [N]} \Omega_i \to \mathbb{R}$ such that it satisfies the Lipschitz condition with bounds $(c_i)_{i \in [N]}$. For independent random variables $X^1, \ldots , X^N$ with $X^i$ taking values in $\Omega_i$, we have that for any $\lambda > 0$
\[
\Pro{f(X^1, \ldots, X^N) \geq \Ex{f(X^1, \ldots, X^N)} + \lambda} \leq \exp\left( - \frac{2 \cdot \lambda^2}{\sum_{i = 1}^N c_i^2} \right).
\]
\end{thm}

\begin{lem}[Azuma's Inequality for Super-Martingales {\cite[Problem 6.5]{DP09}}] \label{lem:azuma}
Let $X^0, \ldots, X^N$ be a super-martingale satisfying $|X^{i} - X^{i-1}| \leq c_i$ for any $i \in [N]$, then for any $\lambda > 0$,
\[
\Pro{X^N \geq X^0 + \lambda} \leq \exp\left(- \frac{\lambda^2}{2 \cdot \sum_{i=1}^n c_i^2} \right).
\]
\end{lem}

In order to state the concentration inequality for supermartingales conditional on a bad event not occurring, we introduce the following definitions from~\cite{CL06}. Consider any random variable $X$ (in our case it will be the $Z^t$, the adjusted quadratic potential in \cref{lem:quadratic_eps}) that can be evaluated by a sequence of decisions $Y^1, Y^2, \ldots ,Y^N$ of finitely many outputs (the allocated balls). We can describe the process by a \textit{decision tree} $T$, a complete rooted tree with depth $n$ with vertex set $V(T)$. Each edge $uv$ of $T$ is associated with a probability $p_{uv}$ depending on the decision made from $u$ to $v$. 

We say $f : V (T) \to \mathbb{R}$ satisfies an \textit{admissible condition} $P$ if $P = \{P_v\}$ holds for every vertex $v$. For an admissible condition $P$, the associated bad set $\mathcal{B}^i$ over the $X_i$ is defined to be
\[
\mathcal{B}^i = \{ v \mid \text{the depth of $v$ is $i$, and $P_u$ does not hold for some ancestor $u$ of $v$} \}.
\]

\begin{thm}[Theorem 8.3 in~\cite{CL06}]
\label{thm:chu_lu_thm_8_3}
For a filtration $\mathbf{F}$,
\[
\{\emptyset, \Omega \} = \mathfrak{F}^0 \subseteq \mathfrak{F}^1 \subseteq \ldots \subseteq \mathfrak{F}^N,
\]
suppose that the random variable $X^i$ is $\mathfrak{F}^i$-measurable for $0 \leq i \leq N$. Let $\mathcal{B} = \mathcal{B}^N$ denote the bad set with the following admissible condition:
\begin{align*}
    \Ex{X^i \mid \mathfrak{F}^{i-1}} & \leq X^{i-1},\\
    |X^i - X^{i-1}| & \leq c_i,
\end{align*}
for $1 \leq i \leq N$ and for $c_1, \ldots , c_N \geq 0$. Then, we have
\[
\Pro{X^N \geq X^0 + \lambda} \leq \exp\left( - \frac{\lambda^2}{2 \cdot \sum_{i = 1}^N c_i^2}\right) + \Pro{\mathcal{B}}.
\]
\end{thm}

We continue with three drift inequalities, whose proofs follow from some straightforward martingale techniques and the Optional Stopping Theorem.
\begin{lem}\label{lem:drift_prob}
Let $X^{t}$ be a stochastic process on the integers $\{ 0,1,\ldots , M\}$, for some finite $M > 0$. Let $D^{t}:=X^{t}-X^{t-1}$, and assume for all $t \geq 0$,
\begin{align*}
    \Ex{ D^{t} \, \left| \, \mathfrak{F}^{t-1}, X^{t-1} > 0 \right. } &\leq 0. %
\end{align*}
Let $X^0 > 0$ and define for some integer $k \geq 1$, $\tau:=\min\{ t \geq 0 \colon X^t = 0 \vee X^t \geq k\}$. Then, \[
\Pro { X^{\tau} = 0} \geq 1 - \frac{X^0}{k}.
\]
\end{lem}
\begin{proof}
First, define $Y^{t}:=X^{t \wedge \tau}$. Then,
\[
 \Ex{ Y^{t} \, \left| \, \mathfrak{F}^{t-1} \right.} \leq Y^{t},
\]
so $Y^t$ is a super-martingale. Further, note that $|Y^{t}| \leq  M$ holds deterministically.
By the Optional Stopping Theorem,
\begin{align*}
    \Ex{ Y^{\tau} } &\leq Y^{0} = X^0,
\end{align*}
but also
\begin{align*}
    \Ex{ Y^{\tau} } &\geq \Pro { X^{\tau} = 0} \cdot 0 + \left(1 - \Pro { X^{\tau} = 0} \right) \cdot k.
\end{align*}
Combining these, yields
$
 \Pro{ X^{\tau} = 0 } \geq 1- \frac{X^0}{k}.
$
\end{proof}

\begin{lem}\label{lem:drift_walk}
Let $X^{t}$ be a stochastic process on the integers $\{0,1,\ldots, 2s\}$, where $X^{0}=s$.
Define $D^{t}:=X^{t}-X^{t-1}$, and further let $\tau:=\min\{ t \geq 0 \colon X^t = 0 \vee X^t = 2s \}$. Assume for all $t \geq 0$,
\begin{align*}
    \Ex{ D^{t} \,\left|\, \mathfrak{F}^{t-1}, X^{t-1} > 0 \right.} &\leq 0, \\
     \Ex{ (D^{t})^2 \,\left|\, \mathfrak{F}^{t-1} \right.} &\geq \sigma^2. %
\end{align*}
Then,
\[
 \Ex{\tau \, \left| \, X^0 = s \right.} \leq \frac{5s^2}{\sigma^2}.
\]
\end{lem}
\begin{proof}
Define $Z^{t}:=(X^{t})^2 + \lambda X^t + \mu \cdot t $, and recall that $D^{t}:=X^{t}-X^{t-1}$. 

We now show that for $\mu := -\sigma^2$ and $\lambda := -4 s$, for any $t-1 < \tau$, $Z^t$ forms a submartingale,
\begin{align*}
 \Ex{ Z^{t}  \, \left| \, \mathfrak{F}^{t-1} \right.} 
 & = \Ex{ \left. (X^{t-1} + D^t)^2 + \lambda (X^{t-1} + D^t) + \mu t  \, \right| \, \mathfrak{F}^{t-1} }  \\
 & = (X^{t-1})^2 + \lambda X^{t-1} + \mu \cdot t + \left( 2 X^{t-1}+ \lambda \right) \cdot \Ex{ D^{t} \, \left| \, \mathfrak{F}^{t-1} \right. }  + \Ex{ (D^{t})^2 \, \left| \, \mathfrak{F}^{t-1} \right.} \\
 & = Z^{t-1} + \mu +  \left( 2 X^{t-1} + \lambda \right) \cdot \Ex{ D^{t} \, \left| \, \mathfrak{F}^{t-1} \right.}  + \Ex{ (D^{t})^2 \, \left| \, \mathfrak{F}^{t-1} \right.} \\
 & \geq Z^t,
\end{align*}
using in the last inequality that $2 X^{t-1} \leq 4s$ (for $t-1 < \tau$), $\Ex{ D^{t} \, \left| \, \mathfrak{F}^{t-1} \right.}$ and $\Ex{ (D^{t})^2 \,\mid\, \mathfrak{F}^{t-1} } \geq \sigma^2$.

The state space of $X$ is finite and using \cref{lem:drift_prob} for $k = 2 x^{\rho}$, $\Pro{X^t = 0 ~\Big|~X^\rho = x^\rho} \geq 1 - \frac{1}{2}$ for any $t \geq \rho$ and any $x^\rho \in [M+1]$, implying that $\ex{\tau} < \infty$. Further note that $|Z^{t+1} - Z^t| \leq M^2 + 4sM + \mu < \infty$. Hence, applying the Optional Stopping Theorem,
\[
    \Ex{ Z^{\tau} } \geq Z^0,
\]
and thus
\begin{align}\label{eq:x_tau_squared}
 \Ex{ (X^{\tau})^2 } - 4s \Ex{ X^{\tau} } - \sigma^2 \Ex{ \tau}   &\geq s^2 -4s^2 = -3s^2.
\end{align}

Further, by \cref{lem:drift_prob}, with $k=2\cdot s$,
\[
 \Ex{ (X^{\tau})^2 } =
 \Pro{ X^{\tau} = 0} \cdot 0^2 + \left(1- \Pro{X^{\tau}= 0} \right) \cdot (2s)^2 \leq 
 \frac{1}{2} \cdot (2s)^2.
\]
Dropping the negative term in the middle of \cref{eq:x_tau_squared} and rearranging, yields,
\[
    \Ex{\tau} \leq \frac{5s^2}{\sigma^2}. \qedhere
\]
\end{proof}

\begin{lem}\label{lem:drift_walk_two}
Let $X^{t}$ be a stochastic process on the integers $\{0,1,\ldots, M\}$, for some finite $M > 0$.
Define $D^{t}:=X^{t}-X^{t-1}$ and assume for all $t \geq 0$,
\begin{align*}
   \Ex{ D^{t} \, \mid \, \mathfrak{F}^{t-1}} &\geq 0 \\
      \Ex{ (D^{t})^2 \, \mid \, \mathfrak{F}^{t-1}} &\geq \sigma^2.
\end{align*}
Let $X^0=s$, and define for some $k$ with $s < k \leq M$, the stopping time $\tau:=\min\{t \geq 0 \colon X^t \geq k\}$. Then,
\[
 \Ex{ \tau \, \left| \, X^0 = s \right.} \leq \frac{ \Ex{ (X^{\tau})^2 \,\left|\, X^0 = s \right.} - s^2}{\sigma^2}.
\]
\end{lem}
\begin{proof}
Similarly to \cref{lem:drift_walk}, define
\[
 Z^t := (X^t)^2 - \sigma^2 \cdot t.
\]
We now show that $Z^t$ is a submartingale for any $t - 1 < \tau$,
\begin{align*}
\Ex{Z^t \mid \mathfrak{F}^{t-1}} & = \Ex{(X^{t-1} + D^t)^2 - \sigma^2 \cdot t \mid \mathfrak{F}^{t-1}} \\
 & = \Ex{(D^t)^2 \mid \mathfrak{F}^{t-1}} + 2 \cdot X^{t-1} \cdot \Ex{D^t \mid \mathfrak{F}^{t-1}} + (X^{t-1})^2 - \sigma^2 \cdot t \\
 & \stackrel{(a)}{\geq} \sigma^2 + 0 + (X^{t-1})^2 - \sigma^2 \cdot t \\
 & = Z^{t-1},
\end{align*}
using in $(a)$ that $\Ex{ D^{t} \, \mid \, \mathfrak{F}^{t-1}} \geq 0$ and $\Ex{ (D^{t})^2 \, \mid \, \mathfrak{F}^{t-1}} \geq \sigma^2$. 

The preconditions imply that $\Pro{D^t > 0 \mid \mathfrak{F}^{t-1}} > \kappa$ for some $\kappa := \kappa(\sigma) > 0$, so $\ex{\tau} < \infty$. Further note that $|Z^{t+1} - Z^t| \leq M^2 + \sigma^2$. Hence, applying the Optional Stopping Theorem,
\begin{align*}
  \Ex{ Z^{\tau} } \geq Z^0 = s^2.
\end{align*}
Hence, 
\begin{align*}
  s^2 \leq \Ex{ \left. (X^{\tau})^2 \,\right|\, X^0 = s} - \sigma^2 \cdot \Ex{\tau \,\left|\, X^0 = s \right.}.
\end{align*}
By re-arranging we get the claim.
\end{proof}

\end{document}